\documentclass[10pt,journal,compsoc]{IEEEtran}
\ifCLASSOPTIONcompsoc
  \usepackage[nocompress]{cite}
\else
  \usepackage{cite}
\fi
\usepackage{ragged2e}
\usepackage{microtype}

\usepackage[utf8]{inputenc}
\input{mysymbol.sty}
\usepackage{amsmath}
\usepackage{amssymb}
\usepackage{amsthm}
\usepackage{bbm}
\usepackage{float,array}

\usepackage{mdframed}
\newmdtheoremenv{boxtheorem}{Theorem}

\usepackage{colortbl}

\usepackage{soul}
\usepackage[font=footnotesize,labelfont=sf,textfont=sf,up, up]{caption}
\usepackage{graphicx} 
\usepackage{hyperref, bigints}

\newtheorem{theorem}{Theorem}
\newtheorem{definition}{Definition}
\newtheorem{assumption}{Assumption}
\newtheorem{proposition}{Proposition}
\newtheorem{lemma}{Lemma}
\newtheorem{corollary}{Corollary}

\newtheorem{remark}{Remark}

\def\diag{\mbox{diag}}

\newcommand{\vertiii}[1]{{\left\vert\kern-0.25ex\left\vert\kern-0.25ex\left\vert#1 \right\vert\kern-0.25ex\right\vert\kern-0.25ex\right\vert}}
\newcommand{\leb}{L^2([0,1])}
\newcommand{\R}{\mathbb{R}}
\newcommand{\vf}{\varphi}
\newcommand{\vv}{\bbv}
\newcommand{\one}{\mathbf{1}}

\newcommand{\degvec}{\bbc^{\mathrm{d}}}
\newcommand{\degfun}{c^{\mathrm{d}}}
\newcommand{\eigvec}{\bbc^{\mathrm{e}}}
\newcommand{\eigfun}{c^{\mathrm{e}}}
\newcommand{\katvec}{\bbc^{\mathrm{k}}}
\newcommand{\katfun}{c^{\mathrm{k}}}

\newcommand{\prfun}{c^{\mathrm{pr}}}
\newcommand{\prob}{\textup{Pr}}
\newcommand{\pagevec}{\bbc^{\mathrm{p}}}

\begin{document}

\title{Centrality measures for graphons: Accounting for uncertainty in networks}

\author{Marco Avella-Medina, Francesca Parise, Michael T. Schaub, and Santiago Segarra,~\IEEEmembership{Member,~IEEE}\thanks{Authors are ordered alphabetically. All authors contributed equally.  M. Avella-Medina is with the Department of Statistics, Columbia University. F. Parise is with the Laboratory for Information and Decision Systems, MIT\@.  M. Schaub is with the Institute for Data, Systems, and Society, MIT and with the Department of Engineering Science, University of Oxford, UK\@. S. Segarra is with the Department of Electrical and Computer Engineering, Rice University. Emails: {\tt\footnotesize marco.avella@columbia.edu, \{parisef, mschaub\}@mit.edu, segarra@rice.edu}. Funding: This work was supported by the Swiss National Science Foundation grants P2EGP1\_168962, P300P1\_177746 (M. Avella-Medina) and P2EZP2\_168812, P300P2\_177805 (F. Parise); the European Union's Horizon 2020 research and innovation programme under the Marie Sklodowska-Curie grant agreement No 702410 (M. Schaub); the Spanish MINECO TEC2013-41604-R and the MIT IDSS Seed Fund Program (S. Segarra).}}
\IEEEtitleabstractindextext{%
\begin{abstract} 
    As relational datasets modeled as graphs keep increasing in size and their data-acquisition is permeated by uncertainty, graph-based analysis techniques can become computationally and conceptually challenging. 
    In particular, node centrality measures rely on the assumption that the graph is perfectly known --- a premise not necessarily fulfilled for large, uncertain networks. 
    Accordingly, centrality measures may fail to faithfully extract the importance of nodes in the presence of uncertainty. 
    To mitigate these problems, we suggest a statistical approach based on graphon theory: we introduce formal definitions of centrality measures for graphons and establish their connections to classical graph centrality measures. 
    A key advantage of this approach is that centrality measures defined at the modeling level of graphons are inherently robust to stochastic variations of specific graph realizations. 
    Using the theory of linear integral operators, we define degree, eigenvector, Katz and PageRank centrality functions for graphons and establish concentration inequalities demonstrating that graphon centrality functions arise naturally as limits of their counterparts defined on sequences of graphs of increasing size. 
    The same concentration inequalities also provide high-probability bounds between the graphon centrality functions and the centrality measures on any sampled graph, thereby establishing a measure of uncertainty of the measured centrality score.
\end{abstract}

}

\maketitle
\IEEEdisplaynontitleabstractindextext%
\IEEEpeerreviewmaketitle%

\IEEEraisesectionheading{\section{Introduction}\label{sec:introduction}}
\IEEEPARstart{M}{any} biological~\cite{Bu03}, social~\cite{Kleinberg99}, and economic~\cite{Garas2010} systems can be better understood when interpreted as networks, comprising a large number of individual components that interact with each other to generate a global behavior.
These networks can be aptly formalized by graphs, in which nodes denote individual entities, and edges represent pairwise interactions between those nodes.
Consequently, a surge of studies concerning the modeling, analysis, and design of networks have appeared in the literature, using graphs as modeling devices.

A fundamental task in network analysis is to identify salient features in the underlying system, such as key nodes or agents in the network. 
To identify such important agents, researchers have developed \emph{centrality measures} in various contexts~\cite{Friedkin1991, Freeman1977,Freeman1978,Bonacich1987,Borgatti2006a}, each of them capturing different aspects of node importance.
Prominent examples for the utility of centrality measures include the celebrated PageRank algorithm~\cite{Page99,Gleich2015}, employed in the search of relevant sites on the web, as well as the identification of influential agents in social networks to facilitate viral marketing campaigns~\cite{Kempe2003}.

A crucial assumption for the applicability of these centrality measures is that the observation of the underlying network is complete and noise free.
However, for many systems we might be unable to extract a complete and accurate graph-based representation, e.g., due to computational or measurement constraints,  errors in the observed data, or because the network itself might be changing over time.
For such reasons,  some recent approaches have considered the issue of robustness of centrality measures~\cite{Costenbader2003,Borgatti2006,Benzi2015,Segarra2016} and general network features~\cite{balachandran2013inference}, as well as their computation in dynamic graphs~\cite{Lerman2010,Pan2011,Grindrod2013}. The closest work to ours is~\cite{Dasaratha2017},  where convergence results are derived for eigenvector and Katz centralities in the context of random graphs generated from a stochastic block model.

As the size of the analyzed systems continues to grow, traditional tools for network analysis have been pushed to their limit. 
For example, systems such as the world wide web, the brain, or social networks can consist of billions of interconnected agents, leading to computational challenges and the irremediable emergence of uncertainty in the observations.
In this context, \textit{graphons} have been suggested as an alternative framework to analyze large networks~\cite{Borgs2017,Lovasz2012}.   
While graphons have been initially studied as limiting objects of large graphs~\cite{Lovasz2006,Borgs2008,Borgs2012}, they also provide a rich non-parametric modeling tool for networks of any size~\cite{Bickel2009,bickel2011method,airoldietal2013,Yang2014}.
In particular, graphons encapsulate a broad class of network models including the stochastic block model~\cite{Holland1983,Snijders1997}, random dot-product graphs~\cite{Young2007}, the infinite relational model~\cite{Kemp2006}, and others~\cite{Goldenberg2010}.
A testament of the practicality of graphons is their use in applied disciplines such as signal processing~\cite{Morency2017}, collaborative learning~\cite{Lee2017}, and control~\cite{Gao2017}.

In this work we aim at harnessing the additional flexibility provided by the graphon framework to suggest a statistical approach to  agents' centralities that inherently accounts for network uncertainty, as detailed next.

\subsection{Motivation}
 Most existing applications of network centrality measures follow the paradigm in Fig.~\ref{fig:motivation}a: a specific graph --- such as a social network with friendship connections --- is observed, and conclusions about the importance of each agent are then drawn based on this graph, e.g., which individuals have  more friends or which have the most influential connections.
Mathematically, these notions of importance  are encapsulated in a centrality measure that ranks the nodes according to the observed network structure.
For instance, the idea that importance derives from having the most friends is captured by degree centrality.
Since the centrality of any node is computed solely from the network structure~\cite{Friedkin1991, Freeman1977,Freeman1978,Bonacich1987,Borgatti2006a}, a crucial assumption hidden in this analysis is that the empirically observed network captures all the data we care about.

However, in many instances in which centrality measures are employed, this assumption is arguably not fulfilled: we typically do not observe the complete network at once.
Further, even those parts we observe contain measurement errors, such as false positive or false negative links, and other forms of uncertainty.
The key question is therefore how to identify crucial nodes via network-based centrality measures without having access to an accurate depiction of the `true' latent network.

One answer to this problem is to adopt a statistical inference-based viewpoint towards centrality measures, by assuming that the observed graph is a specific realization of an underlying stochastic generative process; see Fig.~\ref{fig:motivation}b. 
In this work, in particular, we use graphons to model such underlying generative process, because they provide a rich non-parametric statistical framework.
Further, it has been recently shown that graphons  can  be efficiently estimated from one (or multiple) noisy graph observations~\cite{airoldietal2013,zhangetal2017}. 
Our main contribution is to show that, based on the inferred graphon, one can  compute a latent centrality profile of the nodes that we term \emph{graphon centrality function}.
This graphon centrality may be seen as a fundamental measure of node importance, irrespective of the specific realization of the graph at hand.
This leads to a robust estimate of the centrality profiles of all nodes in the network. 
In fact, we provide high-probability bounds between the distance of such latent graphon centrality functions and the centrality profiles in any realized network, in terms of the network size.

\begin{figure}[tb!]
    \centering
    \includegraphics[width=\columnwidth]{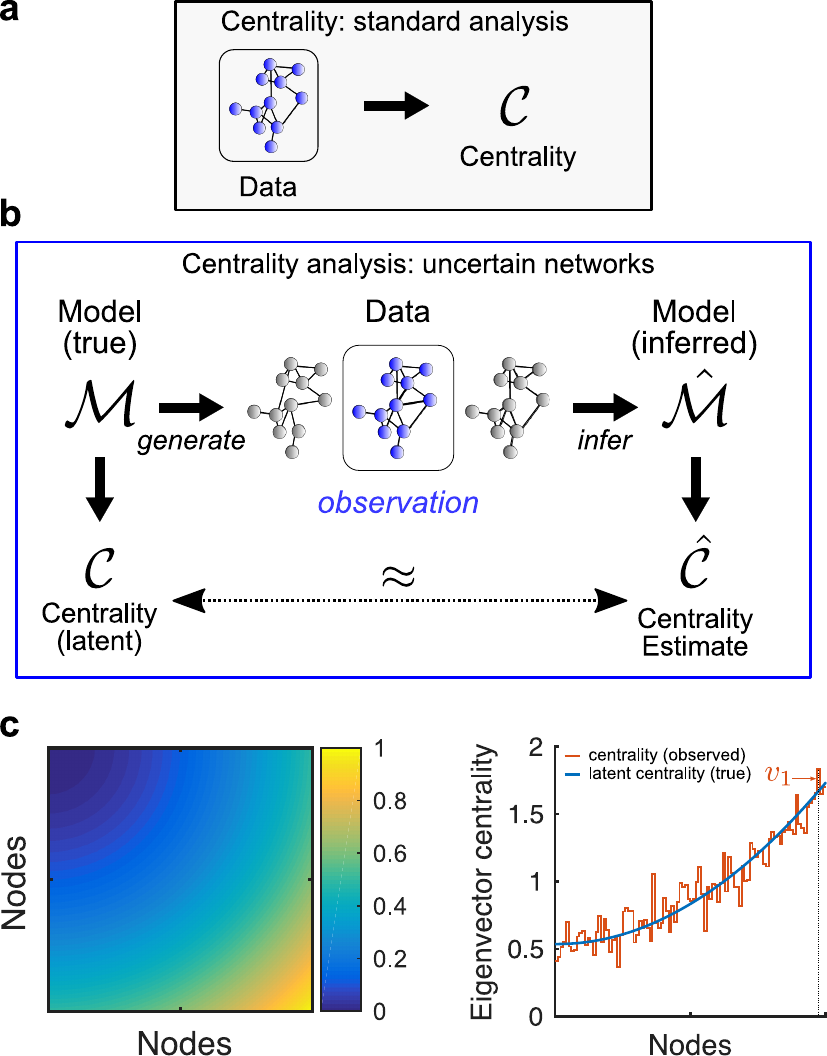}
    \caption{Schematic --- Network centrality analysis. \textbf{(a)} Classical centrality analysis computes a centrality measure purely based on the observed network. \textbf{(b)} If networks are subject to uncertainty, we may adopt a statistical perspective on centrality, by positing that the observed network is but one realization of a true, unobserved latent model. Inference of the model then would lead to a centrality estimate that accounts for the uncertainty in the data in a well-defined manner.
    \textbf{(c)} Illustrative example. 
    Left: A network of 100 nodes is generated according to a graphon model with a well-defined increasing connectivity pattern. Right: This graphon model defines a latent (expected) centrality for each node (blue curve). The centralities of a single realization of the model (red curve) will in general not be equivalent to the latent centrality, but deviate from it. Estimating the graphon-based centrality thus allows us to decompose the observed centrality into an expected centrality score (blue), and a fluctuation that is due to randomness.
}%
    \label{fig:motivation}
\end{figure}

To illustrate the dichotomy of the standard approach towards centrality and the one outlined here, let us consider the graphon in Fig.~\ref{fig:motivation}c. 
Graphons will be formally defined in Section~\ref{Ss:prelim_graphons}, but the fundamental feature is that it defines a random graph model from where graphs of any pre-specified size can be obtained. 
If we generate one of these graphs with $100$ nodes, we can apply the procedure in Fig.~\ref{fig:motivation}a to obtain a centrality value for each agent, as shown in the red curve in Fig.~\ref{fig:motivation}c.
In the standard paradigm we would then sort the centrality values to find the most central nodes, which in this case would correspond to the node marked as $v_1$ in Fig.~\ref{fig:motivation}c.
On the other hand, if we have access to the generative graphon model (or an estimate thereof), then we can compute the continuous graphon centrality function and compare the deviations from it in the specific graph realization; see blue and red curves in Fig.~\ref{fig:motivation}c. 

The result is that while $v_1$ is the most central node in the specific graph realization (see red curve), we would expect it to be less central within the model-based framework since the two nodes to its right have higher latent centrality (see blue curve).
Stated differently, in this specific realization, $v_1$ benefited from the random effects in terms of its centrality. 
If another random graph is drawn from the same graphon model, the rank of $v_1$ might change, e.g., node $v_1$ might become less central. 
Based on a centrality analysis akin to Fig.~\ref{fig:motivation}a, we would conclude that the centrality of this node decreased relative to other agents in the network. 
However, this difference is exclusively due to random variations and thus not statistically significant. 
The approach outlined in Fig.~\ref{fig:motivation}b and, in particular, centrality measures defined on graphons thus provide a statistical framework to analyze centralities in the presence of uncertainty, shielding us from making the wrong conclusion about the change in centrality of node $v_1$ if the network is subject to uncertainty.

A prerequisite to apply the  perspective outlined above is to have a consistent theory of centrality measures for graphons, with well-defined limiting behaviors and well-understood convergence rates.
Such a theory is developed in this paper, as we detail in the next section.

\subsection{Contributions and article structure}
Our contributions are listed below.\\
1) We develop a theoretical framework and definitions for centrality measures on graphons.
Specifically, using the existing spectral theory of linear integral operators, we define the degree, eigenvector, Katz and PageRank centrality functions (see Definition~\ref{def:graphons_centrality}).\\
2) We discuss and illustrate three different analytical approaches to compute such centrality functions (see Section~\ref{sec:computing_centrality}).\\
3) We derive concentration inequalities showing that our newly defined graphon centrality functions are natural limiting objects of centrality measures for finite graphs. These concentration inequalities improve the current state of the art and constitute the main technical results of this paper (see Theorems~\ref{Thm1} and~\ref{Thm2}).\\
4) We illustrate how such bounds can be used to quantify the distance between the latent graphon centrality function and the centrality measures of finite graphs sampled from the graphon.

The remainder of the article is structured as follows.
In Section~\ref{S:preliminaries} we review preliminaries regarding graphs, graph centralities, and graphons.
Subsequently, in Section~\ref{sec:centrality_definitions} we recall the definition of the graphon operator and use it to introduce centrality measures for graphons.
Section~\ref{sec:computing_centrality} discusses how centrality measures for graphons can be computed using different strategies, and provides some detailed numerical examples.
Thereafter, in Section~\ref{S:convergence}, we derive our main convergence results. Section~\ref{S:discussion} provides concluding remarks.
Appendix  A contains proofs omitted in the paper. Appendix B provided in the supplementary material presents some auxiliary results and discussions. 

\textit{Notation:} The entries of a matrix $\mathbf{X}$ and a (column) vector $\mathbf{x}$ are denoted by $X_{ij}$ and $x_i$, respectively; however, in some cases $[\bbX]_{ij}$ and $[\bbx]_i$ are used for clarity.
The notation $^T$ stands for transpose. 
$\diag(\bbx)$ is a diagonal matrix whose $i$th diagonal entry is $x_i$. $\lceil x \rceil$ denotes the ceiling function that returns the smallest integer larger than or equal to $x$.
Sets are represented by calligraphic capital letters, and $1_\ccalB(\cdot)$ denotes the indicator function over the set $\ccalB$. 
$\mathbf{0}$, $\mathbf{1}$, $\bbe_i$, and $\bbI$ refer to the all-zero vector, the all-one vector, the $i$-th canonical basis vector, and the identity matrix, respectively. 
The symbols $\vv$, $\vf$, and $\lambda$ are reserved for eigenvectors, eigenfunctions, and eigenvalues, respectively. 
Additional notation is provided at the beginning of Section \ref{sec:centrality_definitions}.

\section{Preliminaries}\label{S:preliminaries}

In Section~\ref{Ss:prelim_centrality} we introduce basic graph-theoretic concepts as well as the notion of node centrality measures for finite graphs, emphasizing the four measures studied throughout the paper. A brief introduction to graphons and their relation to random graph models is given in Section~\ref{Ss:prelim_graphons}.

\subsection{Graphs and centrality measures}\label{Ss:prelim_centrality}
An undirected and unweighted graph $\ccalG = (\ccalV, \ccalE)$ consists of a set $\ccalV$ of $N$ nodes or vertices and an edge set $\ccalE$ of unordered pairs of elements in $\ccalV$. An alternative representation of such a graph is through its adjacency matrix $\bbA \in \{0 ,1 \}^{N \times N}$, where $A_{ij}=A_{ji}=1$ if $(i,j) \in \ccalE$ and $A_{ij}=0$ otherwise. In this paper  we consider simple graphs (i.e., without self-loops), so that $A_{ii} = 0$ for all $i$.

Node centrality is a measure of the importance of a node within a graph. This importance is not based on the intrinsic nature of each node, but rather on the location that the nodes occupy within the graph. 
More formally, a centrality measure assigns a nonnegative centrality value to every node such that the higher the value, the more central the node is. 
The centrality ranking imposed on the node set $\ccalV$ is in general more relevant than the absolute centrality values. 
Here, we focus on four centrality measures, namely, the degree, eigenvector, Katz and PageRank centrality measures overviewed next; see~\cite{Borgatti2006a} for further details.

\textbf{\emph{Degree centrality}} is a local measure of the importance of a node within a graph. 
The degree centrality  $\degfun_i$ of a node $i$ is given by the number of nodes connected to $i$, that is,
\begin{equation}\label{E:def_degree_centrality}
\degvec := \bbA \one,
\end{equation}
where the vector $\degvec$ collects the values of $c^d_i$ for all $i \in \ccalV$. 

\textbf{\emph{Eigenvector centrality}}, just as degree centrality, depends on the neighborhood of each node. 
However, the centrality measure $\eigfun_i$ of a given node $i$ does not depend only on the number of neighbors, but also on how important those neighbors are. 
This recursive definition leads to an equation of the form $\bbA \eigvec = \lambda \eigvec$, i.e., the vector of centralities $\eigvec$ is an eigenvector of $\bbA$.  
Since $\bbA$ is symmetric, its eigenvalues are real and can be ordered as $\lambda_1 \ge \lambda_2 \geq \ldots \geq \lambda_N$.
The eigenvector centrality $\eigvec$ is then defined as the principal eigenvector $\bbv_1$, associated with $\lambda_1$:
\begin{equation}\label{E:def_eigenvector_centrality}
\eigvec := \sqrt{N} \, \bbv_1.
\end{equation}
For connected graphs, the Perron-Frobenius theorem  guarantees that $\lambda_1$ is a simple eigenvalue, and that there is a unique associated (normalized) eigenvector $\bbv_1$ with positive real entries.
As will become apparent later, the $\sqrt{N}$ normalization introduced in \eqref{E:def_eigenvector_centrality} facilitates the comparison of the eigenvector centrality on a graph to the corresponding centrality measure defined on a graphon.

\textbf{\emph{Katz centrality}} measures the importance of a node based on the number of immediate neighbors in the graph as well as the number of two-hop neighbors, three-hop neighbors, and so on. 
The effect of nodes further away is discounted at each step by a factor $\alpha > 0$. 
Accordingly, the vector of centralities is computed as $\katvec_\alpha = \one + (\alpha \bbA)^1 \one + (\alpha \bbA)^2 \one + \ldots$, where we add the number of k-hop neighbors weighted by $\alpha^k$.
By choosing $\alpha$ such that $0 < \alpha < 1/\lambda_{1}(\bbA)$, the above series converges and we can write the Katz centrality compactly as
\begin{equation}\label{E:def_katz_centrality}
\katvec_\alpha := (\bbI-\alpha \bbA)^{-1}\one.
\end{equation}
Notice that if $\alpha$ is close to zero, the relative weight given to neighbors further away decreases fast, and $\katvec_\alpha$ is driven mainly by the one-hop neighbors just like degree centrality. 
In contrast, if $\alpha$ is close to $1/\lambda_{1}(\bbA)$, the solution to \eqref{E:def_katz_centrality} is almost a scaled version of $\eigvec$. Intuitively, for intermediate values of $\alpha$, Katz centrality captures a hybrid notion of importance by combining elements from degree and eigenvector centralities.
We remark that Katz centrality is sometimes defined as $\katvec_\alpha - \mathbf{1}$. Since a constant shift does not alter the centrality ranking, we here use formula \eqref{E:def_katz_centrality}. We also note that Katz centrality is sometimes referred to as Bonacich centrality in the literature.

\textbf{\emph{PageRank}} measures the importance of a node in a recursive way based on the importance of the neighboring nodes (weighted by their degree). 
Mathematically,
the PageRank centrality of node is given by
\begin{equation}\label{E:def_page_centrality}
\pagevec_\beta := (1-\beta)(\bbI-\beta \bbA \bbD^{-1})^{-1}\one,
\end{equation}
where $0<\beta<1$ 
and $\bbD$ is the diagonal matrix of the degrees of the nodes.
Note that the above formula corresponds to the stationary distribution of a random `surfer' on a graph, who follows the links on the graph with probability $\beta$ and with probability $(1-\beta)$ jumps (`teleports') to a uniformly at random selected node in the graph.
See~\cite{Gleich2015} for further details on PageRank. 

\subsection{Graphons}\label{Ss:prelim_graphons}

A \emph{graphon} is the limit of a convergent sequence of graphs of increasing size, that preserves certain desirable features of the graphs contained in the sequence~\cite{Lovasz2006,Borgs2008,Borgs2012,Lovasz2012,Borgs2017,Orbanz2015,borgs2010moments,lovasz2015automorphism}.
Formally, a graphon is a measurable function $W:[0,1]^2 \to [0,1]$ that is symmetric $W(x, y) = W(y, x)$.
Intuitively, one can interpret the value $W(x,y)$ as the probability of existence of an edge between $x$ and $y$.
However, the `nodes' $x$ and $y$ no longer take values in a finite node set as in classical finite graphs but rather in the continuous interval $[0,1]$. 
Based on this intuition, graphons also provide a natural way of generating random graphs~\cite{Diaconis2007,Orbanz2015}, as introduced in the seminal paper~\cite{Lovasz2006} under the name $W$-random graphs.
In this paper we will make use of the following model, in which the symmetric adjacency matrix $\bbS^{(N)} \in \{0,1 \}^{N \times N}$ of a simple random graph of size $N$ constructed from a graphon is such that for all $i,j\in\{1,\ldots, N\}$
\begin{equation}\label{E:graph_from_graphon}
    \prob[S^{(N)}_{ij} = 1 | u_i, u_j] = \kappa_N W(u_i, u_j),
\end{equation}
where $u_i$ and $u_j$ are latent variables selected uniformly at random from $[0,1]$, and $\kappa_N$ is a constant regulating the sparsity of the graph (see also Definition \ref{D:sampled_graphs})\footnote{Throughout this paper we adopt the terminology of sparse graphs for graphs generated following \eqref{E:graph_from_graphon} with parameter $\kappa_N \to 0$ and  $N\kappa_N  \to \infty$ as $N \to \infty$, even though this does not imply a bounded degree. This terminology is consistent with common usage in the literature \cite{borgs2015private, gao2016optimal,kloppetal2017}. Note also that \cite{veitch2015class} proposed an interesting graph limit framework for graph sequences with bounded degree. }

This means that, when conditioned on the latent variables $(u_1, u_2, \ldots, u_N)$, the off-diagonal entries of the symmetric matrix $\bbS^{(N)}$ are independent Bernoulli random variables with success probability given by $\kappa_NW$. 
In this sense, when $ \kappa_N=1$, the constant graphon $W(x, y) = p$ gives rise to Erd\H{o}s-R\'enyi random graphs with edge probability $p$. 
Analogously, a piece-wise constant graphon gives rise to stochastic block models~\cite{Holland1983,Snijders1997}; for more details see Section \ref{sec:sbm}. 
Interestingly, it can  be shown that the distribution of any simple exchangeable random graph~\cite{Orbanz2015,Goldenberg2010} is characterized by a function $W$ as discussed above~\cite{Aldous1981,Hoover1979,Orbanz2015}.
Finally, observe that  for any measure preserving map $\pi:[0,1]\to [0,1]$, the graphons $W(x,y)$ and $W^\pi(x,y):=W(\pi(x),\pi(y))$ define the same probability distribution on random graphs. 
A precise characterization of the equivalence classes of graphons defining the same probability distribution can be found in \cite[Ch.~10]{Lovasz2012}.

\section{Extending centralities to graphons}
\label{sec:centrality_definitions}

In order to introduce centrality measures for graphons we first introduce a linear integral operator associated with a graphon and recall its spectral properties.
From here on, we denote by $\leb$ the Hilbert function space with inner product $\langle f_1,f_2\rangle:=\int_0^1 f_1(x)f_2(x) \mathrm{d}x$ for $f_1,f_2\in\leb$, and norm $\|f_1\|:=\sqrt{\langle f_1,f_1\rangle}$. 
The elements of $\leb$ are the equivalence classes of Lebesgue integrable functions $f:[0,1]\rightarrow \R$, that is, we identify two functions $f$ and $g$ with each other if they differ only on a set of measure zero (i.e., $f\equiv g \Leftrightarrow \| f - g \| = 0$).  $1_{[0,1]}$ is the identity function in $\leb$.  We use blackboard bold symbols (such as $\mathbb{L}$) to denote linear operators acting on $\leb$, with the exception of $\mathbb{N}$ and $\mathbb{R}$ that denote the sets of natural and real numbers.
The induced (operator) norm is defined as $\vertiii{\mathbb{L}}:= \sup_{f\in\leb \mbox{ s.t. } \|f\|=1} \|\mathbb{L}f\| $.

\subsection{The graphon integral operator and its properties}
Following \cite{Lovasz2012}, we introduce  a linear operator that is  fundamental to derive the notions of centrality for graphons.

\begin{definition}[Graphon operator]
    For a given graphon $W$, we define the associated graphon operator $\mathbb{W}$ as the linear integral operator $\mathbb{W}:\leb\to \leb$ 
$$ f(y) \to (\mathbb{W}f)(x)=\int_0^1W(x,y)f(y)\mathrm{d}y.$$
\end{definition}

From an operator theory perspective, the graphon $W$ is the integral kernel of the linear operator $\mathbb{W}$.
Given the key importance of $\mathbb{W}$, we review its spectral properties in the next definition and lemma.

\begin{definition}[Eigenvalues and eigenfunctions]\label{def:eigen}
    A complex number $\lambda$ is an eigenvalue of $\mathbb{W}$ if there exists a {nonzero} function $\vf \in\leb$, called the eigenfunction, such that
    \begin{equation}\label{E:def_eigenfunction}
    (\mathbb{W}\vf)(x) =\lambda \vf(x).
    \end{equation}
\end{definition}

It follows from the above definition that the eigenfunctions are only defined up to a rescaling parameter. 
Hence, from now on we assume all eigenfunctions are normalized such that $\|\vf\|=1.$

We next recall some known properties of the graphon operator.

\begin{lemma}
\label{lem:graphonoperator}
The graphon operator $\mathbb{W}$ has the following properties.\\[0.1cm]
1) $\mathbb{W}$ is self-adjoint, bounded, and continuous.\\[0.1cm]
    2) $\mathbb{W}$ is diagonalizable.
        Specifically, $\mathbb{W}$ has countably many eigenvalues, all of which are real and can be ordered as $\lambda_1 \ge \lambda_2 \ge \lambda_3 \ge \ldots$.
        Moreover, there exists an orthonormal basis for $\leb$ of eigenfunctions $\{\vf_i\}_{i=1}^\infty$.
        That is,  $(\mathbb{W}\vf_i)(x) =\lambda_i \vf_i(x),$ $\langle \vf_i,\vf_j\rangle=\delta_{i,j}$ for all $i, j$ and any function $f\in \leb$ can be decomposed as $f(x) =\sum_{i=1}^\infty \langle f, \vf_i\rangle \vf_i(x).$
        Consequently,
        $$(\mathbb{W}f)(x)=\sum_{i=1}^\infty \lambda_i  \langle  f, \vf_i\rangle  \vf_i(x).$$
        If the set of nonzero eigenvalues is infinite, then $0$ is its unique accumulation point.
        \\[0.1cm]
    3) Let  $\mathbb{W}^k$ denote $k$ consecutive applications of the operator $\mathbb{W}$.
        Then, for any $k\in\mathbb{N}$,
        $$(\mathbb{W}^kf)(x)=\sum_{i=1}^\infty \lambda_i^k  \langle  f, \vf_i\rangle \vf_i(x).$$\\[0.1cm]
    4) The maximum eigenvalue $\lambda_1$ is positive and there exists an associated eigenfunction $\vf_1$ which is positive, that is, $\vf_1(x)>0$ for all $x\in[0,1]$. Moreover, $\lambda_1= \vertiii{\mathbb{W}}$.

\end{lemma}

Points 1 to 3 of the lemma above can be found in \cite{Borgs2008}, while part 4) follows from the Krein-Rutman theorem \cite[Theorem 19.2]{deimling1985}, upon noticing that the graphon operator $\mathbb{W}$ is positive with respect to the cone $K$ defined by the set of nonnegative functions in $\leb$.

\subsection{Definitions of centrality measures for graphons}\label{Ss:definitions_centralities}
We define centrality measures for graphons based on the graphon operator introduced in the previous section.
These definitions closely parallel the construction of centrality measures in finite graphs; see Section~\ref{Ss:prelim_centrality}.
The main difference is that the linear operator defining the respective centralities is an infinite dimensional operator, rather than a finite dimensional matrix.

\begin{definition}[Centrality measures for graphons]
    \label{def:graphons_centrality}
    Given a graphon $W$ and its associated operator $\mathbb{W}$, we define the following centrality functions:
    \begin{mylist}
   \item  \textbf{1) Degree centrality:} 
            We define $\degfun: [0,1] \to \reals_+$ as
            \begin{equation}\label{E:def_degree_centrality_graphons}
            \begin{aligned}
                c^\mathrm{d}(x)&\textstyle :=(\mathbb{W} 1_{[0,1]})(x)=\int_0^1 W(x,y)\mathrm{d}y.
                \end{aligned}
            \end{equation}
            
       \item \textbf{2) Eigenvector centrality:} 
            For $\mathbb{W}$ with a simple largest eigenvalue $\lambda_1$, let $\vf_1$ ($\|\vf_1\|=1$) be the associated positive eigenfunction.
            The eigenvector centrality function $\eigfun: [0,1] \to \reals_+$ is               %
            \begin{equation}\label{E:def_eigenvector_centrality_graphons}
                c^\mathrm{e}(x) := \vf_1(x). 
            \end{equation}

       \item\textbf{3) Katz centrality:}
            Consider the operator $\mathbb{M}_\alpha$ where
            $(\mathbb{M}_\alpha f)(x):= f(x)-\alpha(\mathbb{W} f)(x)$.
            For any $0<\alpha <1/\vertiii{\mathbb{W}}$, we define the Katz centrality function $\katfun_\alpha: [0,1] \to \reals_+$ as
            \begin{equation}
                c^\mathrm{k}_\alpha(x):=\left(\mathbb{M}_\alpha^{-1}1_{[0,1]}\right)(x).
                \label{eq:katz}
            \end{equation}

             \item \textbf{4) PageRank centrality:} 
            Consider the operator 
            $$(\mathbb{L}_\beta f)(x)=f(x)-\beta \int_0^1W(x,y) f(y) (c^d(y))^\dagger \mathrm{d}y,$$
            where $(c^d(y))^\dagger=(c^d(y))^{-1}$ if $c^d(y)\neq0$ and $(c^d(y))^\dagger= 0$ if $c^d(y)=0$.
            For any $0<\beta <1$, we define $c^\mathrm{pr}_\beta: [0,1] \to \reals_+$ as
            \begin{equation}\label{E:def_pagerank_centrality_graphons}
            \begin{aligned}
                c^\mathrm{pr}_\beta(x)&\textstyle :=(1-\beta) (\mathbb{L}_\beta^{-1} 1_{[0,1]})(x).
                \end{aligned}
            \end{equation}
            Note that $c^d(y)=\int_0^1 W(y, z)\mathrm{d}z = 0$ implies $W(x, y)=0$ almost everywhere.
            
    \end{mylist}

\end{definition}

\begin{remark}\label{R:katz_centrality}
The Katz centrality function is well defined, since for $0<\alpha < 1/ \vertiii{\mathbb{W}}$ the operator $\mathbb{M}_\alpha$ is invertible \cite[Theorem 2.2]{bezandry2011almost}. 
Moreover, denoting the identity operator by $\mathbb{I}$, it follows that $\mathbb{M}_\alpha = \mathbb{I} - \alpha \mathbb{W}$. Hence, by using a Neumann series representation and the properties of the higher order powers of $\mathbb W$ we obtain the equivalent representation
\begin{align*}(\mathbb{M}^{-1}_\alpha f)(x)&\textstyle=((\mathbb{I}-\alpha\mathbb{W})^{-1} f)(x)= \sum_{k=0}^\infty\alpha^k(\mathbb{W}^k f)(x)\\&\textstyle=f(x)+\sum_{k=1}^\infty \alpha^k \sum_{i=1}^\infty\lambda^k_i\langle \vf_i,f\rangle \vf_i(x)\\
&\textstyle=  f(x) + \sum_{i=1}^\infty\frac{\alpha\lambda_i}{1-\alpha\lambda_i}\langle \vf_i,f\rangle \vf_i(x),
\end{align*}
where we used that $|\lambda_i|< \vertiii{\mathbb{W}}$ for all $i$.
Using an analogous series representation it can be shown that PageRank is well defined.
Note also that eigenvector centrality is well-defined by Lemma~\ref{lem:graphonoperator}, part 4).
\end{remark}
Since a graphon describes the limit of an infinite dimensional graph, there is a subtle difference in the semantics of the centrality measure compared to the finite graph setting.
Specifically, in the classical setting the network consists of a finite number of nodes and thus for a graph of $N$ nodes we obtain an $N$-dimensional vector with one centrality value per node.
In the graphon setting, we may think of each real $x \in [0,1]$ as corresponding to one of infinitely many nodes, and thus the centrality measure is described by a function.

\section{Computing centralities on graphons}
\label{sec:computing_centrality}

We illustrate how to compute centrality measures for graphons by studying three examples in detail. 
The graphons we consider are ordered by increasing complexity of their respective eigenspaces and by the generality of the methods used in the computation of the centralities.

\subsection{Stochastic block model graphons}
\label{sec:sbm}
\begin{figure*}[tb!]
	\begin{center}
        \includegraphics{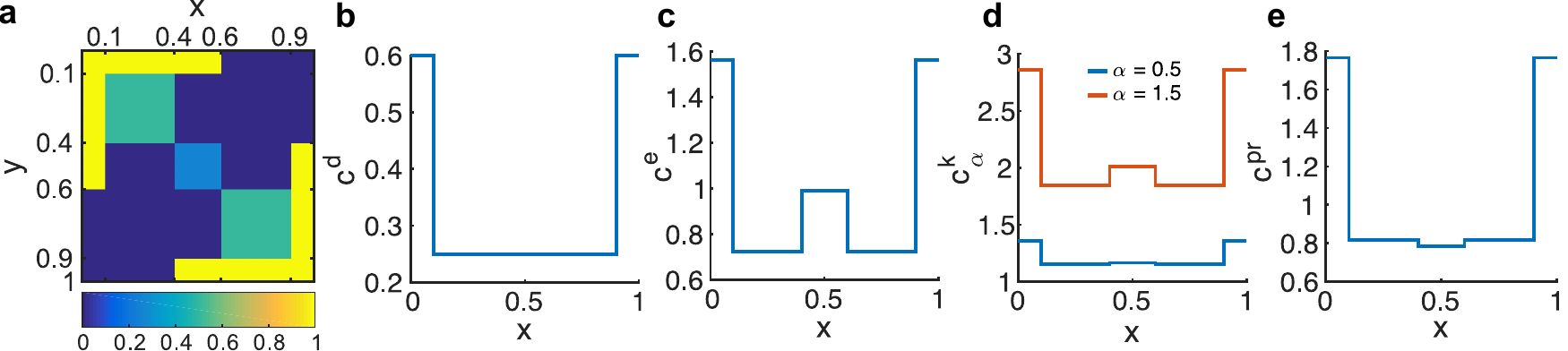}
	\end{center}
	\vspace{-0.3cm}
	\caption{Illustrative example of a graphon with stochastic block model structure. \textbf{(a)} Graphon $W_{\mathrm{SBM}}$ with block model structure as described in the text. \textbf{(b-d)} Degree, eigenvector, Katz, and PageRank centralities for the graphon depicted in (a).}
	\label{fig:example_graphon3}
	\vspace{-0.3cm}
\end{figure*}
We consider a class of piecewise constant graphons that may be seen as the equivalent of a stochastic block model (SBM). Such graphons play an important role in practice, as they enable us to approximate more complicated graphons in a `stepwise' fashion. 
This approximation idea has been exploited to estimate graphons from finite data \cite{airoldietal2013, Latouche2016, wolfeandolhede2013}. In fact, optimal statistical rates of convergence can be achieved over smooth graphon classes \cite{gaoetal2015, kloppetal2017}.
The SBM graphon is defined as follows
\begin{equation}\label{E:def_sbm_graphon}
W_{\mathrm{SBM}}(x,y) :=  \sum_{i=1}^m\sum_{j=1}^mP_{ij} 1_{\mathcal{B}_i}(x)1_{\mathcal{B}_j}(y) ,
\end{equation}
where $P_{ij}\in [0,1]$, $P_{ij}=P_{ji}$, $\cup_{i=1}^m\mathcal{B}_i=[0,1]$ and $\mathcal{B}_i\cap \mathcal{B}_j=\emptyset$ for $i\neq j$.
We define the following $m$ dimensional vector of indicator functions 
\begin{equation}\label{eq:one_vec}
\mathbf{1}(x) := [1_{\mathcal{B}_1}(x), \ldots, 1_{\mathcal{B}_m}(x)]^T,
\end{equation} 
enabling us to compactly rewrite the graphon in \eqref{E:def_sbm_graphon} as
\begin{equation}\label{E:sbm_graphon}
    W_{\mathrm{SBM}}(x,y) = \mathbf{1}(x)^T \mathbf{P} \mathbf{1}(y).
\end{equation}
We also define the following auxiliary matrices.
\begin{definition}
Let us define the {effective measure matrix} $\bbQ_{\mathrm{SBM}} \in \reals^{m \times m}$ and the {effective connectivity matrix} $\bbE_{\mathrm{SBM}} \in \reals^{m \times m}$ for SBM graphons as follows
\begin{equation}\label{E:def_matrix_Q_E}
 \mathbf{Q}_{\mathrm{SBM}} := \int_0^1 \mathbf{1}(x) \mathbf{1}(x)^T \mathrm{d}x, \qquad   \mathbf{E}_{\mathrm{SBM}} := \mathbf{PQ_{\mathrm{SBM}}}.
\end{equation}
\end{definition}
Notice that $\bbQ_{\mathrm{SBM}}$ is a diagonal matrix with entries collecting the sizes of each block. 
Similarly, the matrix $\bbE_{\mathrm{SBM}}$ is obtained by weighting the probabilities in $\bbP$ by the sizes of the different blocks. 
Hence, the effective connectivity from block $\ccalB_i$ to two blocks $\ccalB_j$ and $\ccalB_k$ may be equal even if the latter block $\ccalB_k$ has twice the size ($Q_{kk} = 2 Q_{jj}$), provided that it has half the probability of edge appearance ($2 P_{ik} = P_{ij}$). 
Notice also  that the matrix $\bbE_{\mathrm{SBM}}$ need not be symmetric. 
As will be seen in Section~\ref{Ss:class_finite_rank_graphons}, the definitions in \eqref{E:def_matrix_Q_E} are  specific examples of more general constructions.

The following lemma relates the spectral properties of $\mathbf{E}_{\mathrm{SBM}}$ to those of the operator $\mathbb{W}_{\mathrm{SBM}}$ induced by $W_{\mathrm{SBM}}$.
We do not prove this lemma since it is a special case of Lemma~\ref{lem:finitesum}, introduced in Section~\ref{Ss:class_finite_rank_graphons} and shown in  Appendix A.

\begin{lemma}
\label{lem:sbm}
Let $\lambda_i$ and $\vv_i$ denote the eigenvalues and eigenvectors of  $\mathbf{E}_{\mathrm{SBM}}$ in \eqref{E:def_matrix_Q_E}, respectively. Then, all the nonzero eigenvalues of $\mathbb{W}_{\mathrm{SBM}}$ are given by $\lambda_i$ and the associated eigenfunctions are of the form $\vf_i(x) = \mathbf{1}(x)^T\vv_i$.
\end{lemma}

Using the result above, we can compute the centrality measures for stochastic block model graphons based on the effective connectivity matrix.

\begin{proposition}[Centrality measures for SBM graphons]\label{Prop:centrality_sb} Let $\lambda_i$ and $\vv_i$ denote the eigenvalues and eigenvectors of  $\mathbf{E}_{\mathrm{SBM}}$ in \eqref{E:def_matrix_Q_E}, respectively, and define the diagonal matrix $\bbD_{\bbE} := \mathrm{diag}(\bbE_{\mathrm{SBM}} \mathbf{1})$. 
The centrality functions $\degfun$, $\eigfun$, $\katfun_\alpha$, and $\prfun_\beta$ of the graphon $W_{\mathrm{SBM}}$  can be computed as follows 
\begin{align}\label{E:centralities_sbm}
\degfun(x) & = {\mathbf{1}(x)}^{T} \mathbf{E}_{\mathrm{SBM}} \mathbf{1} ,\ \  \eigfun(x) = \frac{{\mathbf{1}(x)}^{T} \vv_1}{\sqrt{{\vv_1}{^{T}} \bbQ_{\mathrm{SBM}} \vv_1}} , \\ \katfun_\alpha(x) &= {\mathbf{1}(x)}^{T} (\bbI-\alpha \mathbf{E}_{\mathrm{SBM}})^{-1} \mathbf{1}, \nonumber\\
\prfun_\beta (x) & = (1-\beta) \mathbf{1}(x)^T(\bbI - \beta \bbE_{\mathrm{SBM}} \bbD_{\bbE}^{-1})^{-1} \mathbf{1}. \nonumber
\end{align}
\end{proposition}
We next illustrate this result with an example. Its proof is given in Appendix A. 

\subsubsection{Example of a stochastic block model graphon}
\label{Sss:sbm_example}
Consider the stochastic block model graphon $W_{\mathrm{SBM}}$ depicted in Fig.~\ref{fig:example_graphon3}-(a), with corresponding symmetric matrix $\bbP$ [cf.~\eqref{E:def_sbm_graphon}] as in \eqref{eq:PQE}.  Let us define the vector of indicator functions specific to this graphon $\mathbf{1}(x) := [1_{\ccalB_1}(x), \ldots, 1_{ \ccalB_5}(x)]^T$, where the blocks coincide with those in Fig.~\ref{fig:example_graphon3}-(a), that is, $\ccalB_1= [0, 0.1)$, $\ccalB_2= [0.1, 0.4)$, $\ccalB_3= [0.4, 0.6)$, $\ccalB_4= [0.6, 0.9)$, and $\ccalB_5= [0.9,1]$. 
To apply Proposition~\ref{Prop:centrality_sb} we need to compute the effective measure and effective connectivity matrices [cf.~\eqref{E:def_matrix_Q_E}], which for our example are given by
\begin{subequations}\label{eq:PQE} 
    \begin{align}
\bbQ_{\mathrm{SBM}}  \!=\! \diag\begin{pmatrix}
0.1  \\
0.3  \\
0.2  \\
0.3 \\
0.1 \\
\end{pmatrix}\!, \; \bbP\! = 
\begin{bmatrix}
1 & 1 & 1 & 0 & 0 \\
1 & 0.5 & 0 & 0 & 0 \\
1 & 0 & 0.25 & 0 & 1 \\
0 & 0 & 0 & 0.5 & 1 \\
0 & 0 & 1 & 1 & 1 \\
\end{bmatrix} \!\!, 
\end{align}
\begin{align}
\bbE_{\mathrm{SBM}} \!=
\begin{bmatrix}
0.1 & 0.3 & 0.2 & 0 & 0 \\
0.1 & 0.15 & 0 & 0 & 0 \\
0.1 & 0 & 0.05 & 0 & 0.1 \\
0 & 0 & 0 & 0.15 & 0.1 \\
0 & 0 & 0.2 & 0.3 & 0.1 \\
\end{bmatrix}
\!\!.
\end{align}
\end{subequations}
The principal eigenvector of $\bbE_{\mathrm{SBM}}$ is given by $\bbv_1 \approx [0.59, 0.28, 0.38, 0.28, 0.59]^T$. Furthermore, from \eqref{E:centralities_sbm} we can compute the graphon centrality functions to obtain
\begin{align*}
\degfun(x)  & = \mathbf{1}(x) ^T [0.6, 0.25, 0.25, 0.25, 0.6]^T,\\[0.05cm]
\eigfun(x)  & \approx \mathbf{1}(x)^T [1.56, 0.72, 0.99, 0.72, 1.56]^T,\\[0.05cm]
\katfun_\alpha(x)  & \approx 
\begin{cases}
\mathbf{1}(x)^T [1.36, 1.15, 1.16, 1.15, 1.36]^T \quad \text{if $\alpha = 0.5$}, \\
\mathbf{1}(x)^T [2.86, 1.84, 2.01, 1.84, 2.86]^T  \quad \text{if $\alpha = 1.5$},
\end{cases}  \\ 
\prfun_{0.85}(x)  & \approx \mathbf{1}(x) ^T [1.77, 0.82, 0.78, 0.82, 1.77]^T,
\end{align*}
where for illustration purposes we have evaluated the Katz centrality for two specific choices of $\alpha$, and we have set $\beta = 0.85$ for the PageRank centrality. 
These four centrality functions are depicted in Fig.~\ref{fig:example_graphon3}(b)-(e). 

\begin{figure*}[tb!]
	\begin{center}
		\includegraphics[]{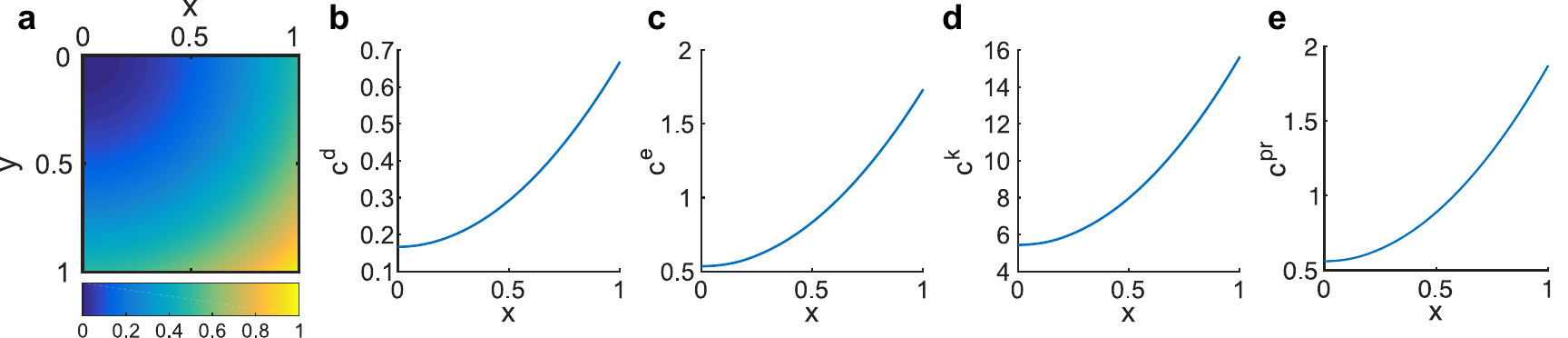}
	\end{center}
	\vspace{-0.3cm}
	\caption{Illustrative example of a graphon with finite rank. \textbf{(a)} The graphon $W_{\mathrm{FR}} = (x^2 +y^2)/2$ is decomposable into a finite number of components, inducing a finite-rank graphon operator. \textbf{(b-d)} Degree, eigenvector, Katz, and PageRank centralities for the graphon depicted in (a).}
	\label{fig:example_graphon2}
		\vspace{-0.3cm}
\end{figure*}

Note that these functions are piecewise constant according to the block partition $\{\mathcal{B}_i\}_{i=1}^m$. Moreover,
as expected from the functional form of $W_{\mathrm{SBM}}$ in Fig.~\ref{fig:example_graphon3}-(a), blocks $\ccalB_1$ and $\ccalB_5$ are the most central as measured by any of the four studied centralities. 
Regarding the remaining three blocks, degree centrality deems them as equally important whereas eigenvector centrality considers $\ccalB_3$ to be more important than $\ccalB_2$ and $\ccalB_4$. 
To understand this discrepancy, notice that in any finite realization of the graphon $W_{\mathrm{SBM}}$, most of the edges from a node in block $\ccalB_3$ will go to nodes in $\ccalB_1$ and $\ccalB_5$, which are the most central ones. 
On the other hand, for nodes in blocks $\ccalB_2$ and $\ccalB_4$, most of the edges will be contained within their own block. 
Hence, even though nodes corresponding to blocks $\ccalB_2$, $\ccalB_3$, and $\ccalB_4$ have the same expected number of neighbors -- thus, same degree centrality -- the neighbors of nodes in $\ccalB_3$ tend to be more central, entailing a higher eigenvector centrality. 
As expected, an intermediate situation occurs with Katz centrality, whose form is  closer to degree centrality for lower values of $\alpha$ (cf. $\alpha = 0.5$) and closer to eigenvector centrality for larger values of this parameter (cf. $\alpha = 1.5$).
For the case of PageRank, on the other hand, block $\ccalB_3$ is deemed as less central than $\ccalB_2$ and $\ccalB_4$. This can be attributed to the larger size of these latter blocks. Indeed, the classical PageRank centrality measure is partially driven by size~\cite{Gleich2015}.

\subsection{Finite-rank graphons}\label{Ss:class_finite_rank_graphons}

We now consider a class of finite-rank (FR) graphons that can be written as a finite sum of products of integrable functions.  Specifically, we consider graphons of the form
\begin{equation}\label{E:def_finite_rank_graphon}
W_{\mathrm{FR}}(x,y) :=\sum_{i=1}^m g_i(x) h_i(y) = \bbg(x)^T \bbh(y), 
\end{equation}
where $m\in\mathbb{N}$ and we have defined the vectors of functions $\bbg(x) = [g_1(x), \ldots, g_m(x)]^T$ and $\bbh(y) = [h_1(y), \ldots, h_m(y)]^T$. 
Observe that $\bbg(x)$ and $\bbh(y)$ must be chosen so that $W_{\mathrm{FR}}$ is symmetric, and $W_{\mathrm{FR}}(x,y) \in [0,1]$ for all $(x,y)\in[0,1]^2$.
Based on $\bbg(x)$ and $\bbh(y)$ we can define the generalizations of $\bbQ_{\mathrm{SBM}}$ and $\bbE_{\mathrm{SBM}}$ introduced in Section~\ref{sec:sbm}, for this class of finite-rank graphons.

\begin{definition}\label{D:def_matrices_Q_E}
	The effective measure matrix $\bbQ$ and the effective connectivity matrix $\bbE$ for a finite-rank graphon $W_{\mathrm{FR}}$ as defined in \eqref{E:def_finite_rank_graphon} are given by
	\begin{equation}\label{E:def_matrix_Q_E_FR}
	\mathbf{Q} := \int_0^1 \bbg(x) \bbg(x)^T \mathrm{d}x,\quad    \mathbf{E} := \int_0^1 \bbh(x) \bbg(x)^T \mathrm{d}x.
	\end{equation}
\end{definition}

The stochastic block model graphon operator introduced in \eqref{E:def_sbm_graphon} is a special case of the class of operators in~\eqref{E:def_finite_rank_graphon}. More precisely, we recover the SBM graphon by choosing $g_i(x) = 1_{\mathcal{B}_i}(x) $ and $h_i(y) = \sum_{j=1}^m P_{ij} 1_{\mathcal{B}_j}(y)$ for $i=1,\dots,m$. The matrices defined in \eqref{E:def_matrix_Q_E} are recovered when specializing Definition~\ref{D:def_matrices_Q_E} to this choice of $g_i(x)$ and $h_i(y)$. We may now relate the eigenfunctions of the FR graphon with the spectral properties  of $\bbE$, as explained in the following lemma.

\begin{lemma}
\label{lem:finitesum}
Let $\lambda_i$ and $\vv_i$ denote the eigenvalues and eigenvectors of  $\mathbf{E}$ in \eqref{E:def_matrix_Q_E_FR}, respectively. Then, all the nonzero eigenvalues of $\mathbb{W}_{\mathrm{FR}}$, the operator associated with \eqref{E:def_finite_rank_graphon}, are given by $\lambda_i$ and the associated eigenfunctions are of the form $\vf_i(x) = \bbg(x)^T \vv_i$.
\end{lemma}

Lemma \ref{lem:finitesum} is proven in Appendix A and shows  that the graphon in \eqref{E:def_finite_rank_graphon} is of finite rank since it has at most $m$ non-zero eigenvalues.
Notice that Lemma~\ref{lem:sbm} follows from Lemma~\ref{lem:finitesum} when specializing the finite rank operator to the SBM case as explained after Definition~\ref{D:def_matrices_Q_E}. Moreover, we can leverage the result in Lemma~\ref{lem:finitesum} to find closed-form expressions for the centrality functions of FR graphons. To write these expressions compactly, we define the vectors of integrated functions $\bbg := \int_0^1 \bbg(y) \mathrm{d}y$ and $\bbh := \int_0^1 \bbh(y) \mathrm{d}y$, as well as the following normalized versions of $\bbh$ and $\bbE$
\begin{equation}\label{E:additional_def_FR_pr}
\bbh_\mathrm{nor} = \int_0^1 \frac{\bbh(y)}{\bbg^T\bbh(y)} \mathrm{d}y, \,\, \bbE_\mathrm{nor} = \int_0^1 \frac{\bbh(y) \bbg(y)^T}{\bbg^T\bbh(y)} \mathrm{d}y.
\end{equation}
With this notation in place, we can establish the following result, which is proven in Appendix A.
\begin{proposition}[Centrality measures for FR graphons]\label{Prop:centrality_fr}
	Let $\vv_1$ be the principal eigenvector  of $\mathbf{E}$ in \eqref{E:def_matrix_Q_E_FR}. 
    Then, the centrality functions $\degfun$, $\eigfun$, $\katfun_\alpha$, and $\prfun_\beta$ of the graphon $W_{\mathrm{FR}}$ can be computed as follows
	\begin{align}\label{E:centralities_fr}
	\degfun(x) = & \bbg(x)^{T}  \bbh, \quad \eigfun(x) = \frac{\bbg(x)^{T} \vv_1}{\sqrt{\vv_1{^{T}} \bbQ \vv_1}} , \\ 
	\katfun_\alpha(x) &= 1 + \alpha \bbg(x)^T \left( \bbI - \alpha \bbE \right)^{-1}\bbh, \nonumber \\
	\prfun_\beta(x) &= (1-\beta) (1+\beta \bbg(x)^T (\bbI - \beta \bbE_\mathrm{nor})^{-1} \bbh_\mathrm{nor}).    \nonumber
	\end{align}
\end{proposition}

In the next subsection we illustrate the use of Proposition~\ref{Prop:centrality_fr} for the computation of graphon centralities.
 
\subsubsection{Example of a finite-rank graphon}
\label{Sss:fr_example}

Consider the FR graphon given by
\begin{equation*}
W_{\mathrm{FR}}(x,y)=(x^2+y^2)/2,
\end{equation*}
and illustrated in Fig.~\ref{fig:example_graphon2}-(a). Notice that this FR graphon can be written in the canonical form \eqref{E:def_finite_rank_graphon} by defining the vectors $\bbg(x) = [x^2, \, 1/2]^T$ and $\bbh(y) = [1/2, \, y^2]^T$. From \eqref{E:def_matrix_Q_E_FR} we then compute the relevant matrices $\bbQ$ and $\bbE$, as well as the relevant normalized quantities, to obtain
\begin{align*}
\bbQ = 
\begin{bmatrix}
1/5 & 1/6 \\
1/6 & 1/4 
\end{bmatrix}, 
\quad
\bbE = 
\begin{bmatrix}
1/6 & 1/4 \\
1/5 & 1/6 
\end{bmatrix}, 
\quad
\bbg = 
\begin{bmatrix}
1/3 \\
1/2 
\end{bmatrix}, \\
\bbh = 
\begin{bmatrix}
1/2 \\
1/3 
\end{bmatrix},\quad 
\bbh_\mathrm{nor} \approx 
\begin{bmatrix}
1.81 \\
0.79 
\end{bmatrix}, \quad
\bbE_\mathrm{nor} \approx
\begin{bmatrix}
0.40 & 0.91 \\
0.40 & 0.40 
\end{bmatrix}.
\end{align*}
\begin{figure*}[tb!]
	\begin{center}
		\includegraphics{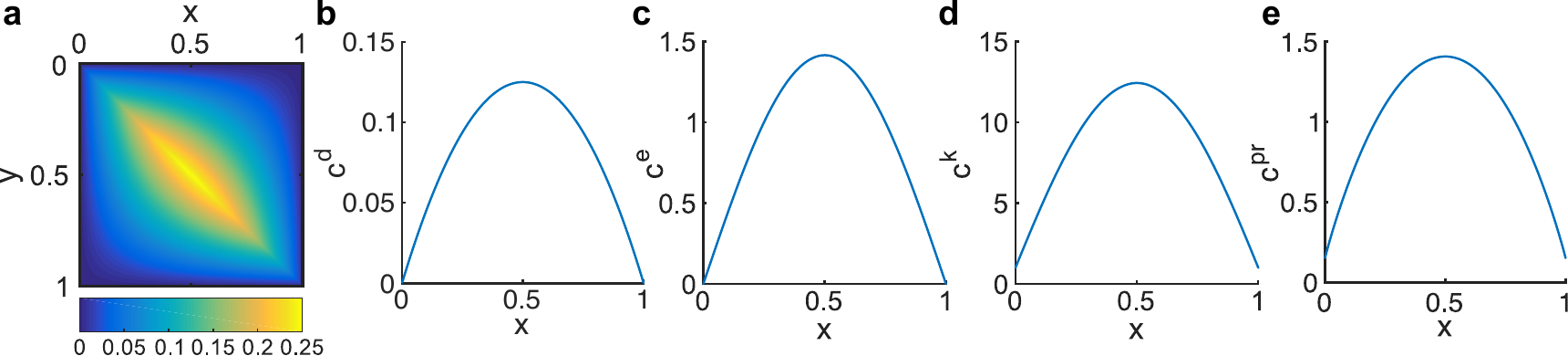}
	\end{center}
	\vspace{-0.3cm}
	\caption{Illustrative example of a general smooth graphon. \textbf{(a)} The graphon $W_{\mathrm{G}}$ induces an operator that has countably infinite number of nonzero eigenvalues and corresponding eigenfunctions. \textbf{(b-d)} Degree, eigenvector, Katz, and PageRank centralities 
    for the graphon depicted in (a).}
	\label{fig:example_graphon1}
		\vspace{-0.3cm}
\end{figure*}
A simple computation reveals that the principal eigenvector of $\bbE$ is $\vv_1 = [\sqrt{10}/3, \,\, 2 \sqrt{2}/3]^T$. 
We now leverage the result in Proposition~\ref{Prop:centrality_fr} to obtain
\begin{align*} 
c^d(x)  &=  [x^2, \, 1/2] \,\, 
\begin{bmatrix}
1/2 \\
1/3 
\end{bmatrix}
= \frac{x^2}{2}+\frac{1}{6},\\ 
c^e(x)  &= \frac{3}{2} \sqrt{\frac{3}{3+\sqrt{5}}}  [x^2, \, 1/2]  \begin{bmatrix}
\sqrt{10}/3 \\
2 \sqrt{2}/3 
\end{bmatrix}
\! \! \approx \! 1.07 \, x^2 + 0.54
,\\
c^k_\alpha(x)  &= [x^2, \, 1/2] \!
\left(\!
\begin{bmatrix}
1 & 0 \\
0 & 1 
\end{bmatrix}
\!-\! \alpha 
\begin{bmatrix}
1/6 & 1/4 \\
1/5 & 1/6 
\end{bmatrix}
\right)^{-1}\hspace{-0.4cm} \alpha 
\begin{bmatrix}
1/2 \\
1/3 
\end{bmatrix}
\!\!+\! 1 \\
& \approx 
10.19 \, x^2 + 5.44,\\
\prfun_{0.85}(x) &\approx 1.31 x^2 + 0.56,
\end{align*}
where we have set $\beta = 0.85$ in the PageRank centrality. Moreover, we have evaluated the Katz centrality for $\alpha = 0.9 /\lambda_1$, where $\lambda_1$ is the largest eigenvalue of $\bbE$. The four centrality functions are depicted in Fig.~\ref{fig:example_graphon2}-(b) through (e). 
As anticipated from the form of $W_{\mathrm{FR}}$, there is a simple monotonicity in the centrality ranking for all the measures considered. More precisely, highest centrality values are located close to $1$ in the interval $[0,1]$, whereas low centralities are localized close to $0$. Unlike in the example of the stochastic block model in Section~\ref{Sss:sbm_example}, all centralities here have the same functional form of a quadratic term with a constant offset.

\subsection{General smooth graphons}

In general, a graphon $W$ need not induce a finite-rank operator as in the preceding Sections~\ref{sec:sbm} and \ref{Ss:class_finite_rank_graphons}. However, as shown in Lemma \ref{lem:graphonoperator}, a graphon always induces a diagonalizable operator with countably many nonzero eigenvalues. In most cases, obtaining the degree centrality function is immediate since it entails the computation of an integral [cf.~\eqref{E:def_degree_centrality_graphons}]. On the other hand, for eigenvector, Katz, and PageRank centralities that depend on the spectral decomposition of $\mathbb{W}$, there is no universal technique available. Nonetheless, a procedure that has shown to be useful in practice to obtain the eigenfunctions $\vf$ and corresponding eigenvalues $\lambda$ of smooth graphons is to solve a set of differential equations  obtained by successive differentiation, when possible, of the eigenfunction equation in \eqref{E:def_eigenfunction}, that is, by considering
\begin{equation}\label{E:differential_equations}
\frac{\mathrm{d}^k}{\mathrm{d} x^k}\int_0^1W(x,y)\vf(y)\mathrm{d}y = \lambda \frac{\mathrm{d}^k \vf(x) }{\mathrm{d} x^k},
\end{equation}
for $k\in\mathbb{N}$. 
In the following section we illustrate this technique on a specific smooth graphon that does not belong to the finite-rank class.

\subsubsection{Example of a general smooth graphon}

Consider the graphon $W_{\mathrm{G}}$ depicted in Fig.~\ref{fig:example_graphon1}-(a) and with the following functional form 
\begin{equation*}
W_{\mathrm{G}}(x,y) = \min(x, y) [1 - \max(x,y)].
\end{equation*}
Specializing the differential equations in \eqref{E:differential_equations} for graphon $W_{\mathrm{G}}$ we obtain
\begin{equation}\label{E:differential_equations_graphon_G}
\frac{\mathrm{d}^k}{\mathrm{d} x^k}\!\! \left[\! (1-x)\!\!\! \int_0^x\!\! y\vf(y)\mathrm{d}y\!+\!x\!\!\int_x^1\!\!\!(1-y)\vf(y)\mathrm{d}y \right] \!=\! \lambda \frac{\mathrm{d}^k \vf(x) }{\mathrm{d} x^k}.
\end{equation}
First notice that without differentiating (i.e. for $k=0$) we can determine the boundary conditions $\vf(0) = \vf(1) = 0$. Moreover, by computing the second derivatives in \eqref{E:differential_equations_graphon_G}, we obtain that $-\vf(x) = \lambda \vf''(x)$. From the solution of this differential equation subject to the boundary conditions it follows that the eigenvalues and eigenfunctions of the operator $\mathbb{W}_{\mathrm{G}}$ are
\begin{equation}\label{E:eigenfunctions_graphon_G}
\lambda_n=\frac{1}{\pi^2n^2} \ \mbox{ and } \ \vf_n(x)= \sqrt{2} \sin(n\pi x) \quad \text{for } n \in\mathbb{N}.
\end{equation}

Notice that $\mathbb{W}_{\mathrm{G}}$ has an infinite --- but countable --- number of nonzero eigenvalues, with an accumulation point at zero.  Thus, $\mathbb{W}_{\mathrm{G}}$ cannot be written in the canonical form for finite-rank graphons \eqref{E:def_finite_rank_graphon}. Nevertheless, having obtained the eigenfunctions we can still compute the centrality measures for $W_{\mathrm{G}}$. For degree centrality, a simple integration gives us
$$ c^d(x) = (1-x)\int_0^x y \, \mathrm{d}y + x \int_x^1 (1-y)\,  \mathrm{d}y = \frac{x (1-x)}{2}. $$
From \eqref{E:eigenfunctions_graphon_G} it follows that the principal eigenfunction is achieved when $n=1$. Thus, the eigenvector centrality function [cf.~\eqref{E:def_eigenvector_centrality_graphons}] is given by
$$ c^e(x) =  \sqrt{2} \sin(\pi x).$$
Finally, for the Katz centrality we leverage Remark~\ref{R:katz_centrality} and the eigenfunction expressions in \eqref{E:eigenfunctions_graphon_G} to obtain
\begin{align*}
c^k_\alpha(x) &= 1+ \sum_{k=1}^\infty \alpha^k \sum_{n=1}^\infty \left( \frac{1}{n^2\pi^2} \right)^k \frac{1-(-1)^n}{\pi n} 2  \sin(n\pi x) \\
 &= 1 +  \sum_{n=1}^\infty \frac{2\alpha }{n^2\pi^2 - \alpha } \frac{1-(-1)^n}{\pi n}  \sin(n\pi x),
\end{align*}
which is guaranteed to converge as long as $\alpha < 1/\lambda_1 = \pi^2$. 
We plot these three centrality functions in Fig.~\ref{fig:example_graphon1}-(b) through (d), where we selected $\alpha = 0.9 \pi^2$ for the Katz centrality. Also, in Fig.~\ref{fig:example_graphon1}-(e) we plot an approximation to the PageRank centrality function $\prfun_\beta$ obtained by solving numerically the integral in its definition.
According to all centralities, the most important nodes within this graphon are those located in the center of the interval $[0,1]$, in line with our intuition. 
Likewise, nodes at the boundary have low centrality values.
Note that while the ranking according to all centrality functions is consistent, unlike in the example in Section~\ref{Sss:fr_example}, here there are some subtle differences in the functional forms. In particular, degree centrality is again a quadratic function whereas the eigenvector and Katz centralities are of sinusoidal form.

\section{Convergence of centrality measures}\label{S:convergence}

In this section we derive concentration inequalities relating the newly defined graphon centrality functions with  standard centrality measures. To this end, we start by noting that  while graphon centralities are  functions, standard centrality measures are vectors. To be able to compare such objects, we first show that there is a one to one relation between any finite graph with adjacency matrix $\bbA$ and a suitably constructed stochastic block model  graphon $W_{\textup{SBM} \mid \bbA}$. Consequently,  any centrality measure of $\bbA$ is in one to one relation with the corresponding  centrality function of  $W_{\textup{SBM} \mid \bbA}$.
To this end, for each $N \in \mathbb{N}$, we define a partition of $[0,1]$ into the  intervals $\mathcal{B}^N_i$, where $\mathcal{B}^N_i=[(i-1)/N, i/N)$ for $i=1,\ldots,N-1$, and $\mathcal{B}^N_N=[(N-1)/N, 1]$.
We denote the associated indicator-function vector by $\mathbf{1}_N(x) := [1_{\mathcal{B}^N_1}(x), \ldots, 1_{\mathcal{B}^N_N}(x)]^T,$ consistently with \eqref{eq:one_vec}.
\begin{lemma}\label{lem:equivalence}
For any adjacency matrix $\bbA\in\{0,1\}^{N\times N}$ define the corresponding  stochastic block model  graphon  as
$$W_{\textup{SBM} \mid \bbA}(x,y)=\sum_{i=1}^N\sum_{j=1}^N A_{ij} 1_{\mathcal{B}^N_i}(x)1_{\mathcal{B}^N_j}(x).$$
 Then the centrality function $c_N(x)$ corresponding to the graphon $W_{\textup{SBM} \mid \bbA}$ is given by
\begin{align*}
    c_N(x)= {\mathbf{1}_N(x)}^{T} \bbc_{A}.
\end{align*}
where $\bbc_{A}$ is  the centrality measures of the graph with rescaled adjacency matrix $\frac{1}{N}\mathbf{A}$.
\end{lemma}
\begin{proof}
The graphon $W_{\textup{SBM} \mid \bbA}$ has the stochastic block model structure described in Section \ref{sec:sbm}, with $N$ uniform blocks. By selecting $m=N$ and $\mathcal{B}_i=\mathcal{B}_i^N$ for each $i\in\{1,\ldots,N\}$, we obtain $\mathbf{Q}_\textup{SBM}= \frac{1}{N} \bbI$. Consequently, the formulas in Proposition~\ref{Prop:centrality_sb} simplify as given in the statement of this lemma.
\end{proof}
\begin{remark}
    Note that the scaling factor $\frac 1N$ does not affect the centrality rankings of the nodes in $\bbA$, but only the magnitude of the centrality measures. This re-scaling   is needed to avoid   diverging centrality measures for  graphs of increasing size. Observe further that $ c_N(x)$ is the piecewise-constant function corresponding to the vector $\bbc_{A}$. We finally remark that graphons of the type $W_{\textup{SBM} \mid \bbA}$ have appeared before in the literature using a different notation~\cite{Borgs2008}.
\end{remark}
By using the previous lemma, we can compare centralities of graphons and graphs by working in the function space. Using this equivalence, we demonstrate that the previously defined graphon centrality functions are not only defined analogously to the centrality measures on finite graphs, but also emerge as the limit  of those centrality measures for a sequence of graphs of increasing size. 
Stated differently, just like the graphon provides an appropriate limiting object for a growing sequence of finite graphs, the graphon centrality functions can be seen as the appropriate limiting objects of the finite centrality measures as the size of the graphs tends to infinity.
In this sense, the centralities presented here may be seen as a proper generalization of the finite setting, just like a graphon provides a generalized framework for finite graphs.
Most importantly, we show that the distance (in $L_2$ norm) between the graphon centrality function and the step-wise constant function associated with the centrality vector of any graph sampled from the graphon can be bounded, with high probability, in terms of the sampled graph size $N$. 

\begin{definition}[Sampled graphon]\label{D:sampled_graphons}
    Given a graphon $W$ and a size $N\in\mathbb{N}$ fix the latent variables $\{u_i\}_{i=1}^N$ by choosing either:
    \begin{mylist}
    \item[-] `deterministic latent variables': $ u_i=\frac iN$.
    \item[-] `stochastic latent variables': $ u_i= U_{(i)}$ where $U_{(i)}$ is the $i$-th order statistic of $N$ random samples from $Unif[0,1]$.
    \end{mylist}
    Utilizing such latent variables  construct
        \begin{mylist}
        \item[-] the `probability' matrix $\bbP^{(N)}\in[0,1]^{N\times N}$ 
            \begin{align*}
                P^{(N)}_{ij} :=\textstyle  W\left(u_i,u_j\right) \ \text{for all} \ i,j\in\{1,\ldots,N\}.
            \end{align*}
        \item[-]  the sampled graphon
            \begin{align*}
                W_N(x,y)& :=\textstyle \sum_{i=1}^N\sum_{j=1}^N P^{(N)}_{ij}1_{\mathcal{B}^N_i}(x)1_{\mathcal{B}^N_j}(y).
            \end{align*}
        \item[-] the operator $\mathbb W_N$ of the sampled graphon
            \begin{align*}(\mathbb{W}_Nf)(x)
                                            &:=\textstyle \sum_{j=1}^N P^{(N)}_{ij} \int_{\mathcal{B}^N_j} f(y)\mathrm{d}y \mbox{ for any } x\in\mathcal{B}^N_i
            \end{align*}
    \end{mylist}
\end{definition}

The sampled graphon $W_N$ obtained when working with deterministic latent variables can intuitively be seen as an approximation of the graphon $W$ by using a stochastic block model graphon with $N$ blocks, as the one described in Section~\ref{sec:sbm}, and is useful to study graphon centrality functions as limit of graph centrality measures. 
On the other hand, the sampled graphon $W_N$ obtained when working with stochastic latent variables is useful as an intermediate step to analyze the relation between the graphon centrality function and the centrality measure of graphs sampled from the graphon according to the following procedure.

\begin{definition}[Sampled graph]\label{D:sampled_graphs}
    Given the `probability' matrix $\bbP^{(N)}$ of a sampled graphon we define
    \begin{mylist}
        \item[-] the sampled matrix $\bbS^{(N)}\in\{0,1\}^{N\times N}$ as the adjacency matrix of a symmetric (random) graph obtained by taking $N$ isolated vertices $i\in\{1,\ldots,N\}$, and adding undirected edges between vertices $i$ and $j$ at random with probability $\kappa_NP^{(N)}_{ij}$ for all $i > j$.
            Note that $\mathbb{E}[\bbS^{(N)}]= \kappa_N\bbP^{(N)}$. 
        \item[-] the associated (random) linear operator
            \begin{align}
                (\mathbb{S}_Nf)(x):=\textstyle \sum_{j=1}^N S^{(N)}_{ij} \int_{\mathcal{B}^N_j} f(y)\mathrm{d}y \mbox{ for any } x\in\mathcal{B}^N_i,
            \end{align}
            and its associated (random) graphon:
            \begin{equation}
                S_N(x,y)=\mathbf{1}_N(x)^T\mathbf S^{(N)}\mathbf{1}_N(y).
            \end{equation}
    \end{mylist}
\end{definition}

In general, we denote the centrality functions associated with the sampled graphon operator $\mathbb{W}_N$ by $c_N(x)$ whereas the centrality functions associated with the operator $\kappa_N^{-1}\mathbb{S}_N$ are denoted by $\hat{c}_N(x)$. 
Note that thanks to Lemma \ref{lem:equivalence}  such centrality functions are in one to one correspondence with the centrality measures of the finite graphs $\bbP^{(N)}$ and $\kappa_N^{-1}\bbS^{(N)}$, respectively. 
Consequently, studying the relation between $c(x)$, $c_N(x)$ and  $\hat{c}_N(x)$ allows us to relate the graphon centrality function with the centrality measure of graphs sampled from the graphon. 
We note that previous works on graph limit theory imply the convergence of $\mathbb{W}_N$ and $\kappa_N^{-1}\mathbb{S}_N$ to the graphon operator  $\mathbb{W}$ \cite{Borgs2012, Szegedy2011,Lovasz2012}. 
Intuitively,  convergence of $c_N(x)$ and  $\hat{c}_N(x)$ to $c(x)$ then follows from continuity of spectral properties of the graphon operator.  
In the following, we make this argument precise and more importantly we provide a more refined analysis by establishing sharp rates of convergence of the sampled graphs, under the following smoothness assumption on $W$. A more detailed discussion of related convergence results can be found in  Appendix B (see supplementary material).

\begin{assumption}[Piecewise Lipschitz graphon]
\label{lipschitz}
There exists a constant $L$ and a sequence of non-overlapping intervals  $\mathcal{I}_k=[\alpha_{k-1},\alpha_k)$ defined by $0=\alpha_0< \dots < \alpha_{K+1} =1$, for a (finite) $K\in\mathbb{N}$, such that for any $k,l$, any set $\mathcal{I}_{kl}=\mathcal{I}_k\times \mathcal{I}_l$ and pairs $(x,y)\in \mathcal{I}_{kl}$, $(x',y')\in \mathcal{I}_{kl}$ we have that 
$$|W(x,y)-W(x',y')|\leq L(|x-x'|+|y-y'|).$$
\end{assumption}
This assumption has also been used in the context of graphon estimation \cite{airoldietal2013,gaoetal2015} and is typically  fulfilled for most of the graphons of interest.

%

\begin{boxtheorem}[Convergence of graphon operators]\label{Thm1}
For a graphon fulfilling Assumption \ref{lipschitz}, it holds with probability $1-\delta'$ that:
\begin{equation}\label{eq:convergenceWN}
\begin{aligned}
    \vertiii{\mathbb{W}_N-\mathbb{W}}&:= \sup_{\|f\|=1} \|\mathbb{W}_Nf-\mathbb{W}f\| \\&  \leq2 \sqrt{({L^2}-K^2)d_N^2 + {K}{d_N}} =: \rho(N),
\end{aligned}
\end{equation}
where $\delta'=0$ and $d_N=\frac1N$ in the case of deterministic latent variables and $\delta'=\delta\in(Ne^{-N/5},e^{-1})$ and ${d_N:=\frac1N +\sqrt{\frac{8\log(N/\delta)}{(N+1)}}}$ in the case of stochastic latent variables.
Moreover, if $N$ is large enough, as specified in Lemma \ref{lemma:max_degree}, then with probability at least $1-\delta-\delta'$
\begin{equation}\label{eq:convergenceSN}
    \vertiii{ \kappa_N^{-1}\mathbb{S}_N-\mathbb{W} }\leq \sqrt{\frac{4\kappa_N^{-1}\log(2N/\delta)}{N}}+\rho(N). 
\end{equation}
In particular, if $\kappa_N = \frac{1}{N^{\tau}}$ with $\tau\in[0,1)$, then:
\begin{equation}\label{eq:convergenceSN2}
   \lim_{N\rightarrow \infty} \vertiii{ \kappa_N^{-1}\mathbb{S}_N-\mathbb{W} }= 0,  \mbox{ almost surely.}
\end{equation}
\end{boxtheorem}

\begin{remark}
    To control the error induced by the random sampling, we derive a concentration inequality for uniform order statistics, as reported next,  that can be of independent interest. 
As this result is not central to the discussion of our paper, we relegate its proof to  Appendix B.
In addition, we make use of a lower bound on the maximum expected degree as reported in Lemma~\ref{lemma:max_degree} and proven in  Appendix B.
\end{remark}

\begin{proposition}
 \label{concentration_all}
Let $U_{(i)}$ be the order statistics of $N$ points sampled from a standard uniform distribution.
Suppose that $N\geq 20$ and $\delta\in(Ne^{-N/5},e^{-1})$. With probability at least $1-\delta$
 $$\Big|U_{(i)}-\frac{i}{N+1}\Big|\leq \sqrt{\frac{8\log(N/\delta)}{(N+1)}} \quad \mbox{for all }  i.$$
 \end{proposition}

\begin{lemma}[]\label{lemma:max_degree}
If  $N$ is such that
\begin{subequations}\label{degree_condition}
\begin{equation}
    2d_N<\Delta^{(\alpha)}_{\textup{MIN}}:=\min_{k\in \{1,\ldots,K+1\}} (\alpha_k-\alpha_{k-1}),
\end{equation}
\begin{equation}
\frac1N \log\left(\frac{2N}{\delta} \right)+d_N(2K+3L) <C^d:=\max_x c^d(x)
\end{equation}
\end{subequations}
then $C^d_N:=\max_i( \sum_{j=1}^N P^{(N)}_{ij}   )\ge \frac49 \log(\frac{2N}{\delta})$.
\end{lemma}

\noindent \textit{Proof of Theorem \ref{Thm1}:} \\
We prove the three statements sequentially.\\[0.2cm]
\textit{Proof of\textbf{~\eqref{eq:convergenceWN}}.}
First of all note that by definition, for any $(x,y) \in \ccalB_i^N \times \ccalB_j^N$ it holds
$W_N(x,y)=W(u_i,u_j),$
but it is not necessarily true that $(u_i,u_j) \in \ccalB_i^N \times \ccalB_j^N$. 
Let us define $k_i,k_j\in\{1,\dots,K+1\}$ such that   the point $(u_i,u_j)$ belongs to the Lipschitz block $\mathcal{I}_{k_ik_j}$, as defined in Assumption \ref{lipschitz} and illustrated in Fig.~\ref{fig:sets}.
We define as $ \mathcal S_{ij}$ the subset of points in  $\ccalB_i^N \times \ccalB_j^N$ that belong to the same Lipschitz block $\mathcal{I}_{k_ik_j}$ as $(u_i,u_j)$. Mathematically,
\begin{equation*}
    \mathcal S_{ij} = \{(x,y) \in \ccalB_i^N \times \ccalB_j^N |  \left(u_{i},u_{j}\right)\! \!\in \mathcal{I}_{k_ik_j} \text{ and } (x,y)\! \in \mathcal{I}_{k_ik_j}  \textstyle \}.
\end{equation*}

\begin{figure}
\begin{center}
    \includegraphics[width=\columnwidth]{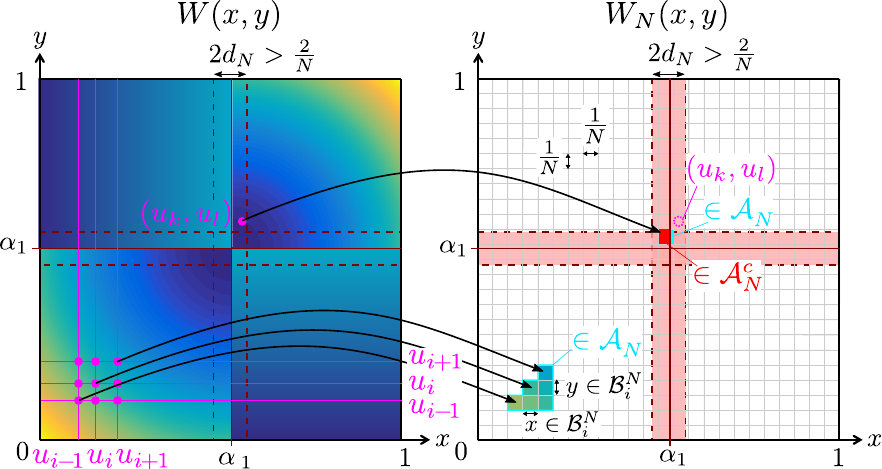}
\end{center}
\vspace{-0.3cm}
\caption{ Schematic for the sets  $\mathcal A_N$ and  $\mathcal A_N^c$ in the case of stochastic latent variables for  $K=1$, threshold $\alpha_1$ and $4$ Lipschitz blocks. 
The plot on the left shows the original graphon $W(x,y)$, the $4$ Lipschitz blocks and some representative latent variables. The plot on the right shows the sampled graphon $W_N(x,y)$ which is a piecewise constant graphon with uniform square blocks of side $\frac1N$. The $\frac1N$-grid is illustrated in gray. The constant value in each block $ \ccalB_i^N \times \ccalB_j^N$ corresponds to the value in the original graph sampled at the point $(u_i,u_j)$ (as illustrated by the arrows). 
The set $\ccalB_i^N \times \ccalB_i^N$ in the bottom left is an example where all the points $(x,y)\in\ccalB_i^N \times \ccalB_i^N$ belong to the same Lipschitz block as their corresponding sample, which is $(u_i,u_i)$, so that  $\ccalB_i^N \times \ccalB_i^N= \mathcal S_{ii}$ and therein $|D(x,y)|\le 2Ld_N$.  The set $\ccalB_k^N \times \ccalB_l^N$ instead is one of the problematic ones since part of its points (cyan) belong to the
same Lipschitz block as $(u_k,u_l)$ and are thus in $\mathcal{S}_{kl}$ but part of its points (red) do not and therefore belong to $\mathcal A_N^c$. Note that by construction with probability $1-\delta'$ all such problematic points are contained in the  set $\mathcal{D}_N$ (which is illustrated in light red). This figure also illustrates that in general $\mathcal{D}_N$ is  a strict superset of $\mathcal A_N^c$ (hence the bound in \eqref{eq:area} is conservative). 
}
\label{fig:sets}
	\vspace{-0.3cm}
\end{figure}

In the following, we partition the set $[0,1]^2$ into the set $\mathcal A_N:=\cup_{ij}  \mathcal S_{ij}$ and its complement $\mathcal A_N^c:=[0,1]^2\backslash \mathcal A_N$. 
In words, $\mathcal A_N$ is the set of points, for which  $(x,y)$  and its corresponding sample $(u_i,u_j)$  belong to the same Lipschitz block.  
We now prove that, with probability $1-\delta'$, $\mathcal A_N^c$ has area 
\begin{equation}\label{eq:area}
    \textstyle   \text{Area}( \mathcal A_N^c)\le \text{Area}(\mathcal D_N) = 4Kd_N-4K^2d_N^2.
\end{equation} 

To prove the above, we define the set $\mathcal D_N$ by constructing a stripe of width $2d_N$ centered at each discontinuity of the graphon, as specified in Assumption \ref{lipschitz}. 
This guarantees that any point in $[0,1]^2\backslash\mathcal D_N$ has distance more than $d_N$ component-wise from a discontinuity. 

Note that in the case of deterministic latent variables for any $i\in\{1,\ldots,N\}$ and any $x\in \ccalB_i^N$ it holds by construction that $|x-u_i| = |x-\frac iN|\le d_N = \frac 1N $ and similarly $|y-u_j|\le d_N$. 
In the case of stochastic latent variables, Proposition~\ref{concentration_all} guarantees that  with probability at least $1-\delta$ for any $i,j\in\{1,\ldots,N\}$ and any $x\in \ccalB_i^N,y\in \ccalB_j^N$   it holds $|x-u_i|=|x-U_{(i)}|\le d_N$, $|y-u_j|\le d_N$. 
In both cases, with probability $1-\delta'$,  all the points in $[0,1]^2\backslash \mathcal D_N$ are less than $d_N$ close to their sample $(u_i,u_j)$ and more than $d_N$ far from any discontinuity (component-wise) hence they surely belong to $\mathcal A_N$. 

Consequently, with probability $1-\delta'$ we have $ \mathcal A_N^c \subseteq \mathcal D_N$.
Each stripe in $\mathcal D_N$ has width $2d_N$, length $1$, and there are $2K$ stripes in total. 
Formula \eqref{eq:area} is then immediate by noticing that  multiplying $2d_N$ times $2K$ counts twice the  $K^2$ intersections between horizontal and vertical stripes.

Consider now any $f\in\leb$ such that $\|f\|=1$. Let $D(x,y):= W_N(x,y)-W(x,y)$ and note that $|D(x,y)|\le1$. Then we get 
\begin{align}
    &\textstyle\|\mathbb{W}_Nf-\mathbb{W}f\|^2=\int_0^1 (\mathbb{W}_Nf-\mathbb{W}f)^2(x) \mathrm{d}x\nonumber \\
                                   &\textstyle= \int_0^1 \left(\int_0^1 D(x,y) f(y) \mathrm{d}y \right)^2  \mathrm{d}x \nonumber\\
                                   &\textstyle\le \int_0^1 \left(\int_0^1 D(x,y)^2 \mathrm{d}y\right) \left(\int_0^1 f^2(y) \mathrm{d}y\right) \mathrm{d}x\label{eq:p1} \\&\textstyle=  \int_0^1 \left(\int_0^1 D(x,y)^2 \mathrm{d}y\right)\|f\|^2  \mathrm{d}x 
                                   \textstyle= \int_0^1 \int_0^1 D(x,y)^2 \mathrm{d}y \mathrm{d}x \nonumber\\
                                   &\textstyle= \iint_{\mathcal A_N}D(x,y)^2\mathrm{d}y\mathrm{d}x   + \iint_{\mathcal A_N^c}D(x,y)^2\mathrm{d}y\mathrm{d}x. \label{eq:p4}
\end{align}
Expression~\eqref{eq:p1} follows from the Cauchy-Schwarz inequality; we used $\|f\|=1$ and, in the last equation, we split the interval $[0,1]^2$ into the sets $\mathcal A_N$ and $\mathcal A_N^c$, as described above and illustrated in Fig.~\ref{fig:sets}.

We can now bound both terms in \eqref{eq:p4}. 
For the first term, note that for all the points $(x,y)$ in $\mathcal A_N$ the corresponding sample $(u_i,u_j)$ belongs to the same Lipschitz block and is at most $d_N$ apart (component-wise). 
Consequently, for these points $|D(x,y)| \le 2 L d_N.$ Overall, we get
$$\textstyle \iint_{\mathcal A_N}D(x,y)^2\mathrm{d}y\mathrm{d}x \le 4L^2d_N^2  \iint_{\mathcal A_N} 1 \,\,  \mathrm{d}y\, \mathrm{d}x  \le 4L^2d_N^2 .$$
For the second term in \eqref{eq:p4}, we use \eqref{eq:area} and the fact that $|D(x,y)|\le 1$ to get
\begin{align*}\textstyle \iint_{\mathcal A_N^c}D(x,y)^2\mathrm{d}y\mathrm{d}x &\textstyle \le  \iint_{\mathcal A_N^c} 1 \,\, \mathrm{d}y \, \mathrm{d}x = \text{Area} (\mathcal A_N^c).
\end{align*}
Substituting these two terms into \eqref{eq:p4} yields 
$$\|\mathbb{W}_Nf-\mathbb{W}f\|^2\le ( 4L^2d_N^2 + 4Kd_N-4K^2d_N^2).$$
Since this  bound holds for all functions $f$ with unit norm, we recover \eqref{eq:convergenceWN}.  \\[0.2cm]

\textit{Proof of \eqref{eq:convergenceSN}.}
 From the triangle inequality we get that
\begin{equation}
\label{eq:B}
\vertiii{ \kappa_N^{-1}\mathbb{S}_N - \mathbb{W}}\leq \vertiii{\kappa_N^{-1}\mathbb{S}_N-\mathbb{W}_N }+ \vertiii{\mathbb{W}_N-\mathbb{W}}.
\end{equation}

We have already bounded the second term on the right hand side of \eqref{eq:B}, so we now concentrate  on the first term.

The operator $\kappa_N^{-1}\mathbb{S}_N-\mathbb{W}_N$ can be seen as the graphon operator of an SBM  graphon with matrix $\kappa_N^{-1}\mathbf{S}^{(N)}-\mathbf{P}^{(N)}$. By Lemma~\ref{lem:sbm} we then have that its eigenvalues coincide with the eigenvalues of the corresponding  $\mathbf{E}_{\mathrm{SBM}}$ matrix which is  $\frac 1N (\kappa_N^{-1}\mathbf{S}^{(N)}-\mathbf{P}^{(N)})$ since in this case $\mathbf{Q}_{\mathrm{SBM}}=\frac1N I_N$ (given that all the intervals $\mathcal{B}^N_i$ have length $\frac1N$).\footnote{Note that Lemma~\ref{lem:sbm} is formulated for graphon operators (i.e. linear integral operators with \emph{nonnegative} kernels). An identical proof shows  that the  result holds also if the kernel assumes negative values.} Consequently,  
\begin{align*}
&\vertiii{\kappa_N^{-1}\mathbb{S}_N-\mathbb{W}_N } = \lambda_\mathrm{max}(\kappa_N^{-1}\mathbb{S}_N-\mathbb{W}_N)\\&=\frac 1N \lambda_\mathrm{max}(\kappa_N^{-1}\mathbf{S}^{(N)}-\mathbf{P}^{(N)}) = \frac{1}{N}\|\kappa_N^{-1}\mathbf{S}^{(N)}-\mathbf{P}^{(N)}\|.
\end{align*}
Hence, to bound the norm of the difference between a random SBM  graphon operator $\kappa_N^{-1}\mathbb{S}_N$ based on the graphon $\kappa_N^{-1}S = \kappa_N^{-1}\mathbf{1}(x)^T \mathbf{S}^{(N)} \mathbf{1}(y)$, with $S^{(N)}_{ij} = {Ber}(\kappa_NP^{(N)}_{ij})$, and its expectation $\mathbb{W}_N$ defined via the graphon $W_N = \mathbf{1}(x)^T \mathbf{P}^{(N)} \mathbf{1}(y)$, we can employ matrix concentration inequalities. Specifically, we use \cite[Theorem 1]{chung2011spectra} in order to bound the deviations of $\|\mathbf{S}^{(N)}-\kappa_N\mathbf{P}^{(N)}\|$.

By Lemma \ref{lemma:max_degree}, for $N$ large enough, the maximum expected degree $\kappa_NC^d_N:=\kappa_N\max_i( \sum_{j=1}^N P^{(N)}_{ij}   )$ of the random graph represented by $\mathbf{S}^{(N)}$ grows at least as $\frac49 \kappa_N\log(\frac{2N}{\delta})$.
Consequently, all the conditions of \cite[Theorem 1]{chung2011spectra} 
are met and we get that with probability $1-\delta -\delta'$ 
\begin{align*}
&\vertiii{\kappa_N^{-1}\mathbb{S}_N-\mathbb{W}_N } = \frac{\kappa_N^{-1}}{N}\|\mathbf{S}^{(N)}-\kappa_N\mathbf{P}^{(N)}\|\\
&\le   \frac{\kappa_N^{-1}}{N}\sqrt{4\kappa_NC^d_N \log(2N/\delta)} \le   \sqrt{\frac{4  \kappa_N^{-1}\log(2N/\delta)}{N}},
\end{align*}
where we used that  $C^d_N\le N$ since each element in $\mathbf{P}^{(N)}$ belongs to $[0,1]$. 

\textit{Proof of \eqref{eq:convergenceSN2}.}
 We finally show that \eqref{eq:convergenceSN} implies almost sure convergence. We start by restating \eqref{eq:convergenceSN} as
\begin{equation}\label{eq:prob}
\textstyle \textup{Pr}\left[ \vertiii{ \kappa_N^{-1}\mathbb{S}_N-\mathbb{W} } \le  \sqrt{\frac{4\kappa_N^{-1}\log(2N/\delta)}{N}} +\rho(N)\right]\ge 1-2\delta.
\end{equation}
Further, pick any $\gamma>0$ and define the infinite sequence of events
 $$\mathcal{E}_N:=\left\{ \vertiii{ \kappa_N^{-1}\mathbb{S}_N-\mathbb{W} } \ge  \gamma +\rho(N)\right\},$$
 for each $N\ge 1$. 
 From \eqref{eq:prob} it follows that
  $\textup{Pr}\left[ \mathcal{E}_N \right]\le 4N \, \textup{exp}\left({-\kappa_NN\gamma^2}/{4}\right).$
  Consequently, if $\kappa_N = \frac{1}{N^{\tau}}$ with $\tau\in[0,1)$, then:
    $$\textstyle \sum_{N=1}^\infty \textup{Pr}\left[ \mathcal{E}_N \right]\le 4 \sum_{N=1}^\infty N\textup{exp}\left(\frac{-\kappa_NN\gamma^2}{4}\right) <\infty$$
and by the Borel-Cantelli lemma there exists  a positive integer $N_\gamma$ such that for all $N\ge N_\gamma$, the complement of $\mathcal{E}_N$, i.e., 
$ \vertiii{ \kappa_N^{-1}\mathbb{S}_N-\mathbb{W} } \le  \gamma +\rho(N) $, holds almost surely. 
To see that $ \vertiii{ \kappa_N^{-1}\mathbb{S}_N-\mathbb{W} } \rightarrow 0 $ almost surely we follow the ensuing argument.
For any given deterministic sequence $\{a_N\}_{N=1}^\infty$ the fact that  for each $\gamma>0$ there is a positive integer $N_\gamma$ such that for all $N\ge N_\gamma$,
$ |a_N| \le  \gamma +\rho(N)$
implies that $a_N\rightarrow 0$.  
In fact for all $\epsilon>0$, if we set $\gamma=\epsilon/2$ and $N_\epsilon:=\max\{N_\gamma, N_\rho\}$ (where $N_\rho$ is the smallest $N$ such that $\rho(N)\le\epsilon/2 $) then  we get that for all $N>N_\epsilon$, $ \vertiii{ \kappa_N^{-1}\mathbb{S}_N-\mathbb{W} }  \le \epsilon$.
Hence,  we can conclude that 
$ \vertiii{ \kappa_N^{-1}\mathbb{S}_N-\mathbb{W} } \rightarrow 0 $ almost surely.
\hfill{$\square$}

The previous theorem provides us with convergence rates for the graphon operators.
Based on these we are able to show a similar convergence result for centrality measures of graphons with a simple dominant eigenvalue.

\begin{assumption}[Simple dominant eigenvalue]\label{bound}
Let  the eigenvalues of $\mathbb{W}$ be ordered such that $\lambda_1 \ge \lambda_2 \ge \lambda_3 \ge \ldots$ and assume that $\lambda_1 > \lambda_2$. 
\end{assumption}

We note that in most empirical studies degeneracy of the dominant eigenvalue is not observed, justifying the above assumption.
A noteworthy exception in which a non-unique dominant eigenvalue may arise is if the graph consists of multiple components.
In this case, however, one can treat each component separately. 

For the proof in case of PageRank we will make the following additional assumption on the graphon.
\begin{assumption}[Minimal degree assumption]\label{ass:degree}
There exists {$\eta>0$} such that  $W(x,y)\ge \eta$ for all  $x,y\in [0,1]$. 
\end{assumption}
Note that while this assumption is not fulfilled, e.g., for a SBM graphon with a block of zero connection probability, it can be further relaxed to accommodate such cases as well.
However, to simplify the proof and avoid additional technical details, we invoke Assumption~\ref{ass:degree} in the following theorem.

\begin{boxtheorem}(Convergence of centrality measures)\label{Thm2}
The following statements hold:

\begin{mylist}
    \item{1)} For any $N\!>\!1$, the centrality functions $c_N(x)$ and $\hat c_N(x)$ corresponding to the operators $\mathbb{W}_N$ and $\kappa_N^{-1}\mathbb{S}_N$, respectively, are in one to one relation with the centrality measures $\bbc_{\bar P^{(N)}}\in\R^N$, $\bbc_{\bar S^{(N)}}\in\R^N$ of the graphs with rescaled adjacency matrices $\bar{\mathbf{P}}^{(N)}:=\frac{1}{N}\mathbf{P}^{(N)}$ and $\bar{\mathbf{S}}^{(N)}:=\frac{1}{N \kappa_N}\mathbf{S}^{(N)}$, via the formula\footnote{Note that $\bar{\mathbf{P}}^{(N)}, \bar{\mathbf{S}}^{(N)}$ belong to $\R_{\ge0}^{N\times N}$  as opposed to $\{0,1\}^{N\times N}$. Nonetheless, the definitions of centrality measures given in Section \ref{Ss:prelim_centrality} can be extended to the continuous interval case in a straightforward manner.}
\begin{align*}
   c_N(x)= {\mathbf{1}_N(x)}^{T} \bbc_{\bar P^{(N)}},  \quad
    \hat c_N(x)= {\mathbf{1}_N(x)}^{T} \bbc_{\bar S^{(N)}}.
\end{align*}
\item{2)} Under  Assumptions \ref{lipschitz}, \ref{bound} and (for PageRank) \ref{ass:degree},  and $N$ sufficiently large, with probability at least $1-\delta'$
$$  \|c_N- c\|\leq C \rho(N) $$
for some constant $C$ and $\rho(N)$, $\delta'$ defined as in Theorem~\ref{Thm1}.
\item{3)} Under  Assumptions \ref{lipschitz}, \ref{bound} and (for PageRank) \ref{ass:degree},  and $N$ sufficiently large, with probability at least $1-2\delta$
\begin{equation*}
\textstyle \|\hat c_N - c\|\leq  C' \bigg(\sqrt{\frac{4\kappa_N^{-1}\log(2N/\delta)}{N}}+\rho(N)\bigg), 
\end{equation*}
for some constant $C'$.
\item{4)} Under  Assumptions \ref{lipschitz}, \ref{bound} and (for PageRank) \ref{ass:degree}, if $\kappa_N = \frac{1}{N^{\tau}}$ with $\tau\in[0,1)$, then:
\begin{equation*}
    \lim_{N\rightarrow \infty} \|\hat c_N - c\|=0, \mbox{ almost surely}.
\end{equation*}
\end{mylist}
\end{boxtheorem}
\begin{proof}
1) Follows immediately from Lemma \ref{lem:equivalence} since  $\mathbb{W}_N$ and $\kappa_N^{-1}\mathbb{S}_N$ are the operators of the graphons corresponding to $\bbP^{(N)}$ and  $\kappa_N^{-1}\bbS^{(N)}$, respectively. 

2) We showed in Theorem~\ref{Thm1} that, under Assumption~\ref{lipschitz}, $\vertiii{\mathbb{W}_N-\mathbb{W}} \leq \rho(N)$ with probability $1-\delta'$. This fact can be exploited to prove convergence of the centrality measures $c_N$ to $c$.  All the subsequent statements hold with probability $1-\delta'$.

\textit{For degree centrality:} $c_N(x)=(\mathbb{W}_N1_{[0,1]})(x)$ and $c(x)=(\mathbb{W}1_{[0,1]})(x)$. Since $\|1_{[0,1]}\|=1$ we get
\begin{equation}
\label{eq:thm2.2a}
\|c_N-c \| = \|(\mathbb{W}_N-\mathbb{W})1_{[0,1]}\| \le \vertiii{\mathbb{W}_N-\mathbb{W}} \leq \rho(N).
\end{equation}

\textit{For eigenvector centrality:}
Let $\{ \lambda_k, \varphi_k\}_{k\ge 1}$, $\{\lambda_k^{(N)},  \varphi_k^{(N)}\}_{k\ge 1}$ be the ordered eigenvalues and eigenfunctions of $\mathbb{W}$ and $\mathbb{W}_N$, respectively. Note that $|{\lambda}_1^{(N)}-\lambda_1|\le \rho(N)$ since
${{\lambda}^{(N)}_1=\vertiii{\mathbb{W}_N}\leq \vertiii{\mathbb{W}}+\vertiii{\mathbb{W}_N-\mathbb{W}}\leq \lambda_1+\rho(N)}$
and 
$\lambda_1=\vertiii{\mathbb{W}}\leq \vertiii{\mathbb{W_N}}+\vertiii{\mathbb{W}_N-\mathbb{W}}\leq {\lambda}^{(N)}_1+\rho(N). $
Furthermore, since by Assumption \ref{bound} we have that $\lambda_1>\lambda_2$,  there exists a large enough $\bar N$ such that for all $N>\bar N$ it holds that $ \lambda_1^{(N)}>\lambda_2$ and $|\lambda_1-\lambda_2|>|{\lambda}^{(N)}_1-\lambda_1|$.
Therefore, by Lemma~8 in Appendix B, we obtain 
\begin{equation}
\label{DKbound}
\|{\varphi}_1^{(N)}-\varphi_1\|\leq \frac{\sqrt{2} \, \vertiii{\mathbb{W}_N-\mathbb{W}} }{|\lambda_1-\lambda_2|-|{\lambda}^{(N)}_1-\lambda_1|}.
\end{equation}
From the facts that $\lambda_1\neq \lambda_2$ (by Assumption \ref{bound}), $|{\lambda}^{(N)}_1-\lambda_1|\leq \rho(N)$, and 
$\vertiii{\mathbb{W}_N-\mathbb{W}} \leq \rho(N)$, it follows that \eqref{DKbound} implies that for $N>\bar{N}$ 
\begin{equation}
\label{tehm2.2b}
\|{\varphi}^{(N)}_1-\varphi_1\|\leq \frac{\sqrt{2}\rho(N)}{|\lambda_1-\lambda_2|-\rho(N)}=O\big(\rho(N)\big) .
\end{equation}

\textit{For Katz centrality:}  Take any value of $\alpha<{1}/{\vertiii{\mathbb{W}}}$, so that $\mathbb{M}_\alpha=\mathbb{I}-\alpha \mathbb{W}$ is invertible and 
$c(x)=\big(\mathbb{M}_\alpha^{-1}1_{[0,1]}\big)(x)$
is well defined.
Since $\vertiii{\mathbb{W}_N-\mathbb{W}} \rightarrow 0$ as $N\rightarrow \infty$, there exists $N_\alpha>0$ such that  $\alpha<{1}/{\vertiii{\mathbb{W}_N}}$ for all $N>N_\alpha$. This implies that for any $N>N_\alpha$,  $[\mathbb{M}_N]_\alpha:=\mathbb{I}-\alpha \mathbb{W}_N$ is invertible and
$c_N(x)=\big([{\mathbb{M}_N}]_\alpha^{-1}1_{[0,1]}\big)(x)$
is well defined. Note that $\vertiii{\mathbb{W}_N-\mathbb{W}}  \leq \rho(N)$ implies
$
 \vertiii{[\mathbb{M}_N]_\alpha-\mathbb{M}_\alpha} =O( \rho(N)).
$
We now prove that 
\begin{equation} \label{eq:m_inv} \quad \vertiii{[\mathbb{M}_N]_\alpha^{-1}-\mathbb{M}_\alpha^{-1}}  = O\big(\rho(N)\big). \end{equation}
To this end, note that $\leb$ is a Hilbert  space, the inverse operator $\mathbb{M}_\alpha^{-1}$ is bounded and  for $N$ large enough it holds $ \vertiii{[\mathbb{M}_N]_\alpha-\mathbb{M}_\alpha}  <1/\vertiii{\mathbb{M}_\alpha^{-1}}$, since $ \vertiii{[\mathbb{M}_N]_\alpha-\mathbb{M}_\alpha} \rightarrow 0$. 
It then follows by \cite[Theorem 2.3.5 ]{atkinson2009theoretical} with  $L:=\mathbb{M}_\alpha, M:=[\mathbb{M}_N]_\alpha$ that
\begin{align*}
\vertiii{[\mathbb{M}_N]_\alpha^{-1}-\mathbb{M}_\alpha^{-1}} \le\! \frac{\vertiii{ \mathbb{M}_\alpha^{-1}}^2 \vertiii{[\mathbb{M}_N]_\alpha-\mathbb{M}_\alpha}}{1\!-\!\vertiii{ \mathbb{M}_\alpha^{-1}} \vertiii{[\mathbb{M}_N]_\alpha-\mathbb{M}_\alpha}}\!=\! O(\rho(N)), 
\end{align*}
thus proving \eqref{eq:m_inv}. 
Finally, since $\|1_{[0,1]}\|=1,$ 
\begin{align*}
    \|c_N-c\|&=\|[{\mathbb{M}_N}]_\alpha^{-1}1_{[0,1]}-\mathbb{M}_\alpha^{-1}1_{[0,1]} \| \\&\le  \vertiii{[\mathbb{M}_N]_\alpha^{-1}-\mathbb{M}_\alpha^{-1}}  
    =O\big(\rho(N)\big).
\end{align*}

\textit{For PageRank centrality:} 
Consider $\beta \in (0,1)$ such that $\mathbb L_\beta$ is invertible and $c^{\mathrm{pr}}(x)$ is well defined. 
Similar to the argument used to show \eqref{eq:m_inv} in the proof for Katz centrality, it suffices to show that $\vertiii{[\mathbb L_N]_\beta-\mathbb L_\beta}=O(\rho(N))$. To show this  note that under Assumption \ref{ass:degree} it holds 
\begin{equation}\label{key_step}
\begin{aligned}
c^d(x)&=\int_0^1 W(x,y)dy \ge \eta, \\
c^d_N(x)&=\int_0^1 W_N(x,y)dy \ge \eta.
\end{aligned}
\end{equation}
Hence $(c^\textup{d}(x))^\dagger=(c^\textup{d}(x))^{-1}\le \frac{1}{\eta}$ and $(c_N^\textup{d}(x))^\dagger=(c_N^\textup{d}(x))^{-1}\le \frac{1}{\eta}$.
For any $f\in\leb$ such that $\|f\|=1$,
\begin{align}\label{eq:page_rank1}
&\| [\mathbb L_N]_\beta f-\mathbb L_\beta f\| = \, \beta \, \| \mathbb W_N\left( f\cdot(c_N^\textup{d})^{-1}\right) - \mathbb W\left( f\cdot(c^\textup{d})^{-1}\right) \| \nonumber \\
&\le \!  \beta \| \!(\mathbb W_N-\mathbb W)\!\!\left( f\cdot\!(c^\textup{d})^{-1}\right)\!\|\!+\! \beta\|\mathbb W_N\!\!\left( f\cdot\!(c_N^\textup{d})^{-1}\!\!\!- f\cdot\!(c^\textup{d})^{-1}\!\right)\! \|\!  \nonumber\\
&\le \, \beta \, \vertiii{ \mathbb W_N-\mathbb W}\| (c^\textup{d})^{-1}\|\!+\beta \vertiii{\mathbb W_N}\| (c_N^\textup{d})^{-1}\!\!- \!(c^\textup{d})^{-1} \|  \nonumber\\
&\le  \, \beta \rho(N) \| (c^\textup{d})^{-1}\| +2\beta \vertiii{\mathbb W}\| (c_N^\textup{d})^{-1}- (c^\textup{d})^{-1} \|, 
\end{align}
where we used that for $N$ large enough $\vertiii{\mathbb W_N}\le 2\vertiii{\mathbb W}$.
In \eqref{eq:page_rank1}, the notation $(c^\textup{d})^{-1}$ is used to denote the function that takes values $(c^\textup{d}(x))^{-1}$ and similarly for $(c_N^\textup{d})^{-1}$.
Observe that equation \eqref{key_step} implies
\begin{equation}\label{eq:page_rank2}
\| (c^\textup{d})^{-1}\| \le \frac{1}{\eta}
\end{equation}
and
 \begin{align}\label{eq:page_rank3}
&\| (c_N^\textup{d})^{-1}- (c^\textup{d})^{-1} \|^2=\int_0^1 \left((c_N^\textup{d}(y))^{-1}- (c^\textup{d}(y))^{-1}\right)^2dy\nonumber
\\&=\int_0^1 \left( \frac{c_N^\textup{d}(y)- c^\textup{d}(y)}{c_N^\textup{d}(y)c^\textup{d}(y)}\right)^2dy \le \frac{1}{\eta^4} \| c_N^\textup{d}- c^\textup{d} \|^2.
\end{align}
Combining \eqref{eq:thm2.2a}, \eqref{eq:page_rank1}, \eqref{eq:page_rank2} and \eqref{eq:page_rank3} yields the desired result
\begin{equation}\begin{aligned}
&\| \mathbb [L_N]_\beta f-\mathbb L_\beta f\| \le \frac{\beta}{\eta} \left[1   +\frac{2}{\eta} \vertiii{\mathbb W} \right]\rho(N) =O(\rho(N)).
\end{aligned}
\end{equation}

 3) It suffices to mimic the argument made above for $\|c_N-c\|$ adapting it for the case $\| \hat{c}_N - c\|$ by making use of \eqref{eq:convergenceSN}. The proof is omitted to avoid redundancy.

 4) By Theorem \ref{Thm1} we have that  $\vertiii{\kappa_N^{-1}\mathbb{S}_N-\mathbb{W}}\rightarrow 0$ almost surely. 
 This means that the set of realizations $\{\tilde{\mathbb{S}}_N\}_{N=1}^\infty$ of $\{\mathbb{S}_N\}_{N=1}^\infty$ for which $\vert\kern-0.25ex\vert\kern-0.25ex\vert{\kappa_N^{-1}\tilde{\mathbb{S}}_N-\mathbb{W}}\vert\kern-0.25ex\vert\kern-0.25ex\vert\rightarrow 0$ has probability one. 
 For each of these realizations it can be proven (exactly as in part 2)\footnote{In part 2(b) the rate of convergence of  $\mathbb{W}_N$ to $\mathbb{W}$ was used. 
 Nonetheless, the same statement holds under the less stringent condition $\vertiii{\mathbb{W}_N-\mathbb{W}}\rightarrow 0$, since this is sufficient to prove that $\lambda_1^{(N)} \rightarrow \lambda_1$.} that 
 $$  \lim_{N\rightarrow \infty} \|\tilde c_N - c\|=0,  $$
 where $\tilde c_N(x)$ is the deterministic sequence of centrality measures associated with the realization $\{\tilde{\mathbb{S}}_N\}_{N=1}^\infty$. 
 Consequently, $ \textup{Pr}[ \lim_{N\rightarrow \infty} \|\hat c_N - c\|=0] =1$ and   ${\lim_{N\rightarrow \infty} \|\hat c_N - c\|=0}$ almost surely.
\end{proof}
To sum up, Theorem \ref{Thm2} shows that, on the one hand, the centrality functions of the finite-rank operators  $\mathbb{W}_N$ and  $\kappa_N^{-1}\mathbb{S}_N$ can be computed by simple interpolation of the centrality vectors of the corresponding finite-size graphs with adjacency matrices $\mathbf{P}^{(N)}$ and  $\kappa_N^{-1}\mathbf{S}^{(N)}$ (suitably rescaled). On the other hand, such centrality functions $c_N(x)$ and $\hat c_N(x)$ become better approximations of the centrality function $c(x)$ of the graphon $W$ as $N$ increases.
As alluded  above, the importance of this result derives from the fact that it establishes that the centrality functions here introduced are the appropriate limits of the finite centrality measures, thus validating the presented framework. We finally note that as immediate corollary of the previous theorem we get the following robustness result for the centrality measures of different realizations.
\begin{corollary}
Consider two  graphs $\bbS^{N_1}$ and $\bbS^{N_2}$ sampled from a graphon $W$ satisfying Assumptions~\ref{lipschitz} and~\ref{bound}. 
Assume without loss of generality that $N_1\le N_2$ and let $\bbc_{S^{(N_i)}}\in\R^N$ be the centrality of the graphs with rescaled adjacency matrices $\bar{\mathbf{S}}^{(N_i)}:=\frac{1}{N_i\kappa_{N_i}}\mathbf{S}^{(N_i)}$, for $i\in\{1,2\}$.
Then for $N$ sufficiently large, with probability at least $1-4\delta$
\begin{align*}
    \left\|{\mathbf{1}_{N_1}(x)}^{T} \bbc_{\bar{\mathbf{S}}^{(N_1)}}\right. &-\left. {\mathbf{1}_{N_2}(x)}^{T} \bbc_{\bar{\mathbf{S}}^{(N_2)}}\right\| \\
                                                                                                                             &\leq  2C' \left(\sqrt{\frac{4\kappa_{N_1}^{-1}\log(2N_1/\delta)}{N_1}}+\rho(N_1)\right)
\end{align*}
for some constant $C'$ and $\rho(N)$, $\delta$ defined as in Theorem \ref{Thm1}.
\end{corollary}
\begin{figure}[tb!]
	\begin{center}
        \includegraphics{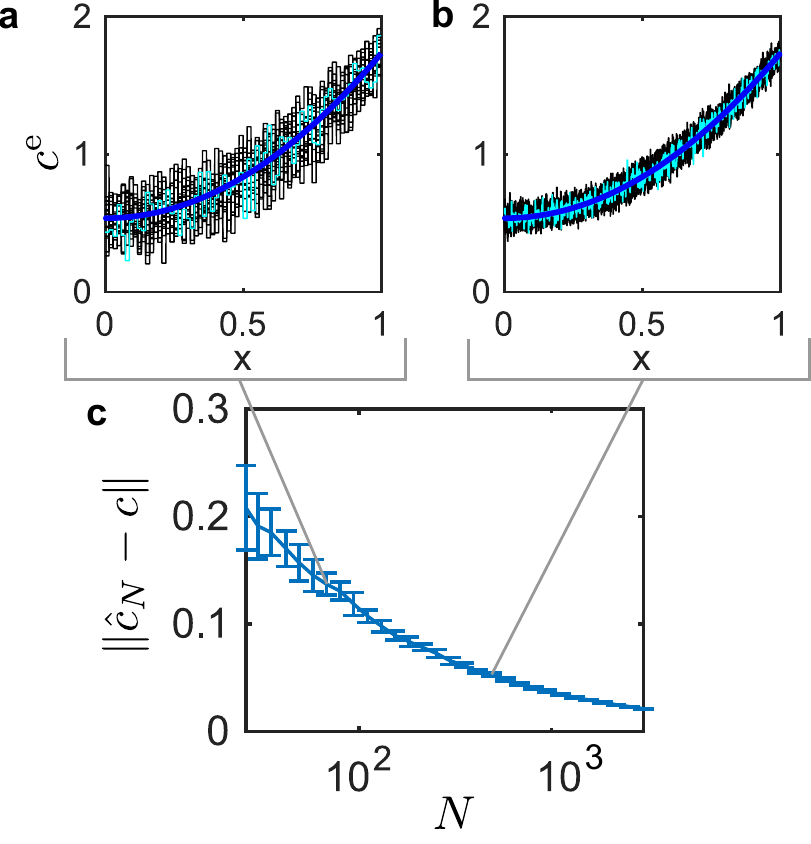}
	\end{center}
	\vspace{-0.3cm}
    \caption{Convergence of the eigenvector centrality function for the FR graphon in section~\ref{Sss:fr_example}. \textbf{(a-b)} Eigenvector centrality functions computed from a sampled graphon with deterministic latent variables (black; one realization shown in cyan for visualization purposes) and the eigenvector centrality function of the continuous graphon (blue). The examples shown correspond to a resolution of (a) $N=68$ grid points, and (b) $N=489$ grid points.
        In each case 20 realizations were drawn from the discretized graphon $\bbP^{(N)}$.
    \textbf{(c)} Convergence of the error $\|\hat{c}_N - c \|$ as a function of the number of grid points $N$, corresponding to the number of nodes in the sampled graph.
            For each point we plot the sample mean $\pm$ one standard deviation. }
	\label{fig:convergence1}
	\vspace{-0.3cm}
\end{figure}

\begin{figure}[tb!]
	\begin{center}
        \includegraphics{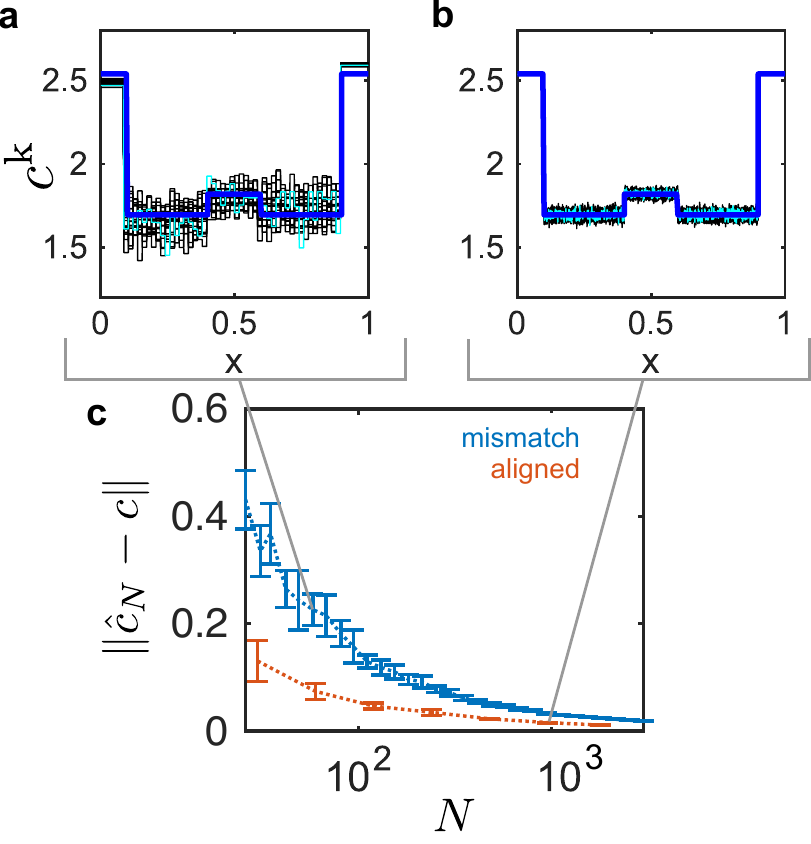}
	\end{center}
	\vspace{-0.3cm}
	\caption{Convergence of the Katz centrality function for the SBM graphon in section~\ref{Sss:sbm_example}. \textbf{(a-b)} Katz centrality functions computed from a sampled graphon with deterministic latent variables (black; one realization shown in cyan for visualization purposes) and the Katz centrality function of the continuous graphon (blue). The examples shown correspond to a resolution of (a) $N=58$ grid points, and (b) $N=960$ grid points.
        In each case 20 realizations were drawn from the discretized graphon $\bbP^{(N)}$.
    \textbf{(c)}~Convergence of the error $\|\hat{c}_N - c \|$ as a function of the number of grid points $N$, corresponding to the number of nodes in the sampled graph. 
        For each point we plot the sample mean $\pm$ one standard deviation.
    For visualization purposes we connect the data-points (red) in which the grid approximation was aligned with the piece-wise constant changes of the SBM (i.e. each grid point is within one Lipschitz block of the SBM graphon -- c.f. Fig.~\ref{fig:sets}). For this graphon, this happens for all $N$ that are divisible by 10.
Likewise, we connected those data points where there was a mismatch between the sampling grid and the SBM graphon structure (blue).}
	\label{fig:convergence2}
	\vspace{-0.3cm}
\end{figure}

To check our analytical results, we performed numerical experiments as illustrated in Figs.~\ref{fig:convergence1} and \ref{fig:convergence2}.
In Fig.~\ref{fig:convergence1}, we consider again the finite-rank graphon from our example in Section~\ref{Sss:fr_example} and assess the convergence of the eigenvector centrality function $\hat{c}_N^\text{e}$ from the sampled networks (with deterministic latent variables and $\kappa_N=1$), to the true underlying graphon centrality measure $c^\text{e}$.
As this graphon is smooth we have $K=0$, i.e., there is only a single Lipschitz block, and we observe a smooth decay of the error when increasing the number of grid points $N$, corresponding to the number of nodes in the sampled graph.

For the stochastic block model graphon $W_\text{SBM}$ from our example in Section~\ref{Sss:sbm_example}, however, we have $K=4$ and thus 25 Lipschitz blocks, which are delimited by discontinuous jumps in the value of the graphon.
The effect of these jumps is clearly noticeable when assessing the convergence of the centrality measures, as illustrated in Fig.~\ref{fig:convergence2} for the example of Katz centrality.
If the deterministic sampling grid of the discretized graphon $\mathbb W_N$ is aligned with the structure of the stochastic block model $\mathbb W_\text{SBM}$, there is no mismatch introduced by the sampling procedure and thus the approximation error of the centrality measure $c_N$ is smaller (also in the sampled version $\hat{c}_N$).
Stated differently, if the $\frac{1}{N}$-grid is exactly aligned with the Lipschitz blocks of the underlying graphon, we are effectively in a situation in which the area $\mathcal A_N^c$ is zero, which is analogous to the case of $K=0$ (see Fig.~\ref{fig:sets}).
In contrast, if there is a misalignment between the Lipschitz blocks and the $\frac{1}{N}$-grid, then additional errors are introduced leading to an overall slower convergence, which is consistent with our results above.

\section{Discussion and Future Work}\label{S:discussion}

In many applications of centrality-based network analysis, the system of interest is subject to uncertainty.
In this context, a desirable trait for a centrality measure is that the relative importance of agents should be impervious to random fluctuations contained in a particular realization of a network.
In this paper, we formalized this intuition by extending the notion of centrality to graphons.
More precisely, we proposed a departure from the traditional concept of centrality measures applied to deterministic graphs in favor of a graphon-based, probabilistic interpretation of centralities.
Thus, we
1) introduced suitable definitions of centrality measures for graphons,
2) showed how such measures can be computed for specific classes of graphons,
3) proved that the standard  centrality measures defined for graphs of finite size converge to our newly defined graphon centralities, and
4) bound the distance between graphon centrality function and centrality measures over sampled graphs.

The results presented here constitute a first step towards a systematic analysis of centralities in graphons and several questions remain unanswered.
In particular, we see two main challenges that need to be addressed to widen the scope of  applications of our methods. 
First, in most practical scenarios the graphon will need to be estimated from (finite) data. 
The validity (error terms) of the centrality scores will accordingly be contingent on the errors made in this estimation~\cite{airoldietal2013,gaoetal2015}. It would therefore be very interesting to explore how to best levarage existing graphon estimation results in order to estimate our proposed graphon centrality functions. 
Second, the parameter $\kappa_N$ allows for the analysis of sparse networks as introduced in \cite{Bickel2009,bickel2011method,borgs2015private, gao2016optimal,kloppetal2017} but not for  networks with asymptotically bounded degrees~\cite{veitch2015class}. 
An extension of our results to the analysis of networks with finite degrees is thus of future interest.

Additionally, there are a number of other generalizations that are worth investigating.

First, the generalization of centralities for graphons beyond the four cases studied in this paper. 
In particular, the extension to centralities that do not rely directly on spectral graph properties, such as closeness and betweenness centralities, appears to be a challenging task. 
Indeed, a suitable notion of path for infinite-size networks needs to be defined first and bounds on the possible fluctuations of such a notion of path would need to be derived.
Enriching the class of graphon centrality measures would also contribute to the characterization of the relative robustness of certain classes of centralities, which would allow us to better assess their utility for the analysis of uncertain networks in practice.

Second, the identification of classes of graphons (others than the ones here discussed) for which explicit and efficient formulas of the centrality functions can be derived. 

Third, the determination of whether the convergence rates provided in Theorem~\ref{Thm2} are also optimal.
A related question in this context is to derive convergence results for the exact ordering of the nodes.
This is in particular relevant for applications where we would like to know by how much the ranking of an individual node might have changed as a result of the uncertainty of the network.
One possible avenue to tackle this kind of question would be to start investigating $\ell_\infty$-norm bounds~\cite{Fan2016} for centralities, which enable us to control the maximal fluctations of each individual entry of the centrality measures.

Finally, the extension of the centrality definitions from regular graphons to more complex objects, such as asymmetric networks or time-varying graphons~\cite{crane2016,pensky2016}. 

\appendices
\section*{Appendix A : Omitted proofs}

\subsection*{Proof of Proposition~\ref{Prop:centrality_sb}}
 
\begin{proof}

This proof is a consequence of Proposition \ref{Prop:centrality_fr}. We notice that we can specialize the formula therein to the case of block stochastic models to obtain the relations:  $\bbh=\bbE_{\textup{SBM}}\mathbf{1}, \bbQ=\bbQ_{\textup{SBM}}, \bbE=\bbE_{\textup{SBM}}, \bbE_{\textup{norm}}=\bbE_{\textup{SBM}} \bbD_{\bbE}^{-1}, \bbh_\textup{nor}=\bbE_{\textup{SBM}} \bbD_{\bbE}^{-1} \mathbf{1} $. 
To see that these equivalencies are true one can check that, for instance:
\begin{align*}
[ \bbh_\textup{nor}]_i&=\textstyle \sum_{j=1}^m \frac{P_{ij} [\bbQ_{\textup{SBM}}]_{jj}}{\sum_{k=1}^m  [\bbQ_{\textup{SBM}}]_{kk} P_{kj}}
= \sum_{j=1}^m \frac{P_{ij} [\bbQ_{\textup{SBM}}]_{jj}}{\sum_{k=1}^m  [\bbQ_{\textup{SBM}}]_{kk} P_{jk}}\\&= \textstyle \sum_{j=1}^m \frac{[\bbE_{\textup{SBM}} ]_{ij}}{[\bbD_{\bbE}]_{jj}}= \textstyle \sum_{j=1}^m [\bbE_{\textup{SBM}} \bbD_{\bbE}^{-1} ]_{ij}.
\end{align*}
From the above equivalences, the formulas for $\degfun(x), \eigfun(x) $ follow immediately. For $\katfun_\alpha(x)$ we obtain
\begin{align*}
\katfun_\alpha(x)&=\textstyle 1+\alpha {\mathbf{1}(x)}^{T}  \left(\sum_{k=0}^\infty \alpha^k  \bbE_{\textup{SBM}}^{k} \right)\bbE_{\textup{SBM}} \mathbf{1}\\
&=\textstyle 1+ {\mathbf{1}(x)}^{T}  \left(\sum_{k=1}^\infty \alpha^k  \bbE_{\textup{SBM}}^{k} \right) \mathbf{1}\\
&=\textstyle  {\mathbf{1}(x)}^{T}  \left(\sum_{k=0}^\infty \alpha^k  \bbE_{\textup{SBM}}^{k} \right) \mathbf{1}={\mathbf{1}(x)}^{T}  \left(1- \alpha  \bbE_{\textup{SBM}} \right)^{-1} \mathbf{1}.
\end{align*}
Finally, $\prfun_\beta(x)$ can be proven similarly.
\end{proof}
 
 \subsection*{Proof of Lemma \ref{lem:finitesum}}
 \begin{proof} 

Assume that $\bbv$ is an eigenvector of $\bbE$ such that $\bbE \bbv = \lambda \bbv$ with $\lambda\neq 0$. We now show that this implies that $\vf(x) = \sum_{j=1}^m v_j g_j(x)$ is an eigenfunction of $\mathbb{W}$ with eigenvalue $\lambda$. From an explicit computation of $(\mathbb{W}\vf)(x)$ we have that
\begin{align*}
(\mathbb{W}\vf)(x) &\textstyle= \sum_{i=1}^m g_i(x) \int_0^1 h_i(y) \sum_{j=1}^m v_j g_j(y) \mathrm{d}y \\&\textstyle= \sum_{i=1}^m g_i(x) \sum_{j=1}^m  v_j \int_0^1 h_i(y)   g_j(y) \mathrm{d}y.
\end{align*}
Recalling the definition of $\bbE$ from \eqref{E:def_matrix_Q_E_FR}, it follows that
\begin{equation*}
(\mathbb{W}\vf)(x)\textstyle\! = \! \sum_{i=1}^m g_i(x) \sum_{j=1}^m  v_j E_{ij}\! = \! \sum_{i=1}^m g_i(x) \lambda v_i \!=\!  \lambda \vf(x),
\end{equation*}
where we used the fact that $\bbv$ is an eigenvector of $\bbE$ for the second equality. 

In order to show the converse statement, let us assume that $\vf$ is an eigenfunction of $\mathbb{W}$ with associated eigenvalue $\lambda \neq 0$. Then, we may write that
    \begin{equation*}
        (\mathbb{W}\vf)(x) \textstyle= \sum_{i=1}^m g_i(x) \int_0^1 h_i(y) \vf(y) \mathrm{d}y = \lambda \vf(x),
    \end{equation*}
    from where it follows that
    \begin{equation}\label{E:proof_finitesum_020}
        \vf(x)\textstyle\! =\! \lambda^{-1} \!\sum_{i=1}^m g_i(x)\!  \int_0^1 h_i(y) \vf(y) \mathrm{d}y \! =\! \sum_{i=1}^m g_i(x)  v_i,
    \end{equation}
    where we have implicitly defined $v_i := \lambda^{-1} \int_0^1 h_i(y) \vf(y) \mathrm{d}y$. Substituting \eqref{E:proof_finitesum_020} into this definition yields, for all $i=1, \ldots, m$,
    \begin{equation*}
        \lambda v_i\textstyle = \int_0^1 h_i(y) \sum_{j=1}^m g_j(y)  v_j \mathrm{d}y = \sum_{j=1}^m v_j  E_{ij} = [\bbE \bbv]_i.
    \end{equation*}
    By writing the above equality in vector form we get that $\bbE \bbv = \lambda \bbv$, thus completing the proof.
\end{proof}

\subsection*{Proof of Proposition~\ref{Prop:centrality_fr}}
\begin{proof}
	The proof for degree centrality follows readily from \eqref{E:def_degree_centrality_graphons}, i.e. $\degfun(x) \! = \! \int_0^1 \bbg(x)^T \bbh(y) \mathrm{d}y = \bbg(x)^T \int_0^1 \bbh(y) \mathrm{d}y = \bbg(x)^T \bbh$. Based on Lemma~\ref{lem:finitesum}, for $\eigfun$ it is sufficient to prove that $\|\eigfun\|=1$. To this end, note that
	\begin{align*}
	\|\eigfun\|^2&=\textstyle \frac{1}{\vv_1{^{T}} \bbQ \vv_1} \int_0^1 (\bbg(x)^{T} \vv_1)^2 \mathrm{d}x \\&\textstyle= \frac{1}{\vv_1{^{T}} \bbQ \vv_1} \int_0^1(\vv_1^T \bbg(x)) (\bbg(x)^{T} \vv_1) \mathrm{d}x\\
	 &\textstyle=\frac{1}{\vv_1{^{T}} \bbQ \vv_1} \vv_1^T \int_0^1 \bbg(x)\bbg(x)^{T}  \mathrm{d}x  \vv_1 =\frac{\vv_1{^{T}} \bbQ \vv_1}{\vv_1{^{T}} \bbQ \vv_1} =1.
	\end{align*}

For Katz centrality, we first prove by induction that 
\begin{equation}\label{E:proof_prop_fr_010}
\textstyle(\mathbb{W}^k_{\mathrm{FR}}f)(x) = \int_0^1 \bbg(x)^T \bbE^{k-1} \bbh(y) f(y) \mathrm{d}y,
\end{equation}
for all finite-rank operators $\mathbb{W}_{\mathrm{FR}}$. The equality holds trivially for $k=1$. Now suppose that it holds for $k-1$, we can then compute
\begin{align*}
&\textstyle(\mathbb{W}^k_{\mathrm{FR}}f)(x)  = (\mathbb{W}_{\mathrm{FR}} \mathbb{W}^{k-1}_{\mathrm{FR}} f)(x)\\&\textstyle = 
\int_0^1 \bbg(x)^T \bbh(z) \int_0^1 \bbg(z)^T \bbE^{k-2} \bbh(y) f(y) \,\mathrm{d}y \,\mathrm{d}z\\
&\textstyle = \int_0^1 \bbg(x)^T \left( \int_0^1 \bbh(z) \bbg(z)^T \,\mathrm{d}z \right) \bbE^{k-2} \bbh(y) f(y) \,\mathrm{d}y 
\\&\textstyle = \int_0^1 \bbg(x)^T \bbE^{k-1} \bbh(y) f(y) \,\mathrm{d}y,
\end{align*}
where we used Definition~\ref{D:def_matrices_Q_E} for the last equality. We leverage \eqref{E:proof_prop_fr_010} to compute $\katfun_\alpha$ using the expression in Remark~\ref{R:katz_centrality},
\begin{align*}
&\textstyle\katfun_\alpha(x)  = 1 + \sum_{k=1}^\infty\alpha^k(\mathbb{W}^k 1_{[0,1]})(x) \\&\textstyle= 1 + \sum_{k=1}^\infty\alpha^k \int_0^1 \bbg(x)^T \bbE^{k-1} \bbh(y) \mathrm{d}y \\&\textstyle= 1 + \sum_{k=1}^\infty\alpha^k \bbg(x)^T \bbE^{k-1} \bbh  = 1 + \alpha \sum_{k=0}^\infty\alpha^k \bbg(x)^T \bbE^{k} \bbh \\&\textstyle= \!1\! +\! \alpha \bbg(x)^T \!\!\left( \sum_{k=0}^\infty \alpha^k  \bbE^{k}\! \right) \!\bbh \!= \!1 \!+\! \alpha \bbg(x)^T\!\! \left( \bbI - \alpha \bbE \right)^{-1}\!\bbh,
\end{align*}
as we wanted to show.

Finally, the methodology to prove the result for PageRank is similar to the one used for Katz, thus we sketch the proof to avoid redundancy. 
First, we recall the definition of $\mathbb{G}$ from the proof of Proposition~\ref{Prop:centrality_sb} and use induction to show that for every finite-rank graphon $(\mathbb{G}^k \mathbf{1}_{[0,1]})(x) = \bbg(x)^T \bbE_\mathrm{nor}^{k-1} \bbh_\mathrm{nor}$ for all integer $k \geq 1$. 
We then compute the inverse in the definition of PageRank~\eqref{E:def_pagerank_centrality_graphons} via an infinite sum as done above for Katz but using the derived expression for $(\mathbb{G}^k \mathbf{1}_{[0,1]})(x)$.
\end{proof}

\bibliographystyle{IEEEtran}
\bibliography{library_graphons.bib}

\begin{IEEEbiography}[{\includegraphics[width=1in,height=1.25in,clip,keepaspectratio]{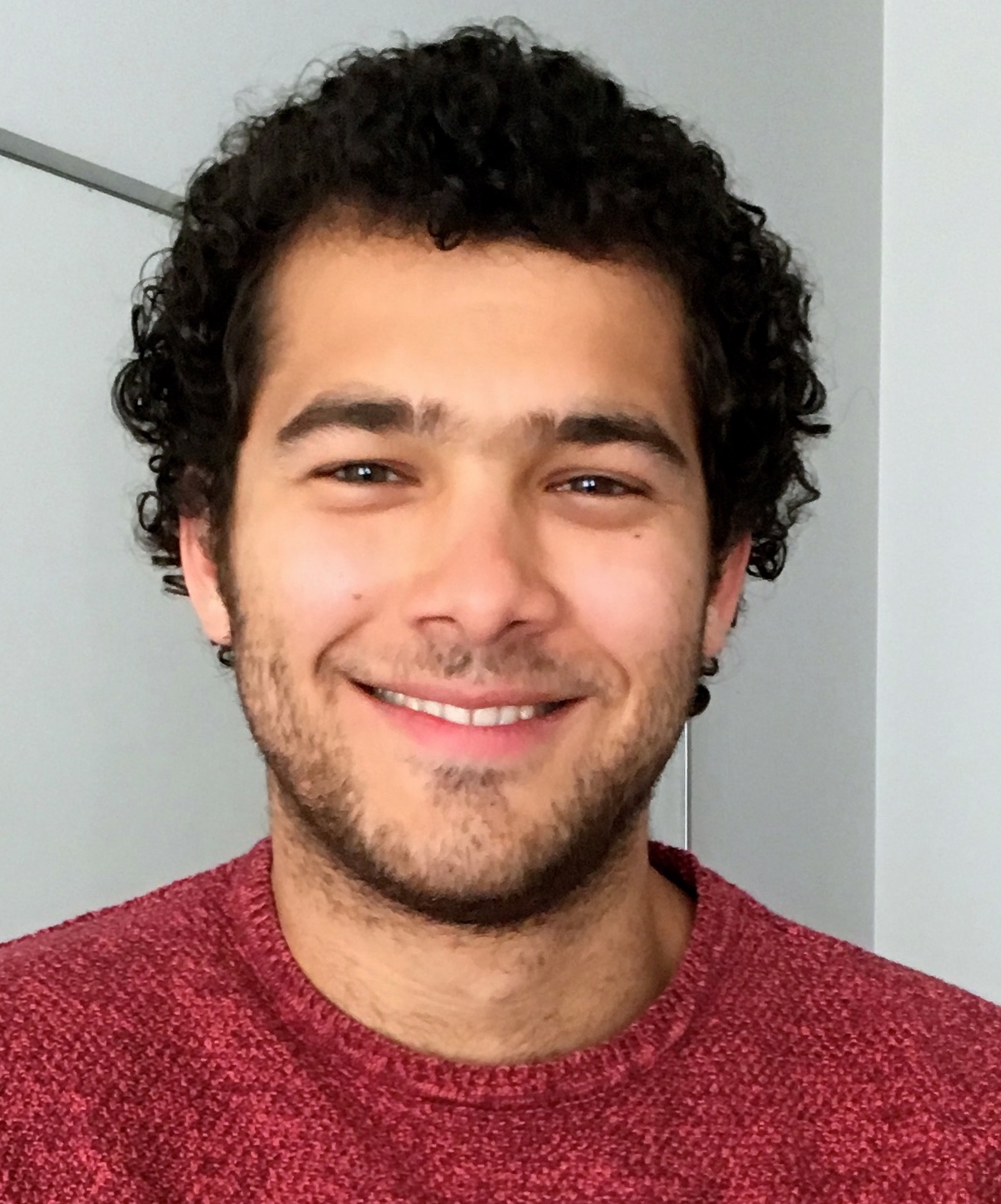}}]{Marco Avella-Medina} holds a B.A. degree in Economics (2009), a M.Sc. degree in Statistics (2011) and a Ph.D. in Statistics (2016) from the University of Geneva, Switzerland. From August 2016 to June 2018 he was a postdoctoral researcher with the Statistics and Data Science Center, and the Sloan School of Management at the Massachusetts Institute of Technology. Since July 2018, he is an Assistant Professor in the Department of Statistics at Columbia University. His research interests include robust statistics, high-dimensional statistics  and  statistical machine learning.
\end{IEEEbiography}
\begin{IEEEbiography}[{\includegraphics[width=1in,height=1.25in,clip,keepaspectratio]{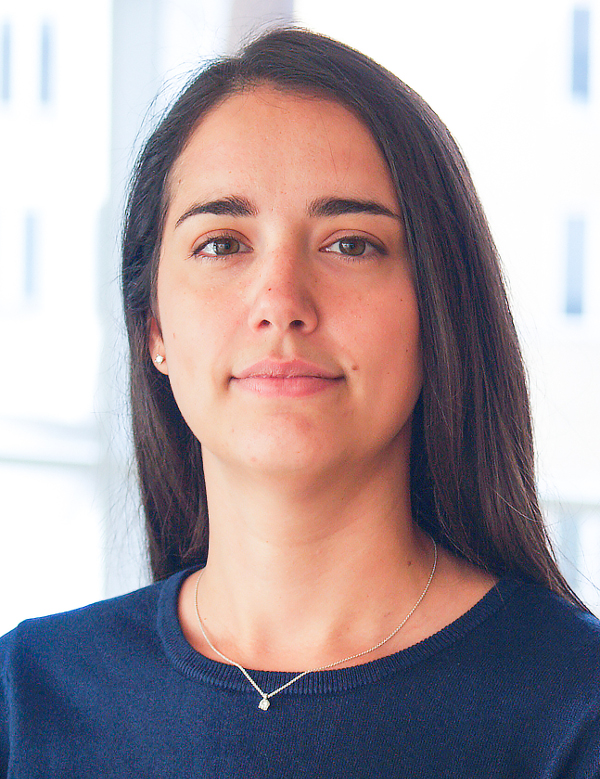}}]{Francesca Parise}
was born in Verona, Italy, in 1988. She received the B.Sc. and M.Sc. degrees (cum Laude) in Information and Automation Engineering from the University of Padova, Italy, in 2010 and 2012, respectively. She conducted her master thesis research at Imperial College London, UK, in 2012. She graduated from the Galilean School of Excellence, University of Padova, Italy, in 2013. She defended her PhD at the Automatic Control Laboratory, ETH Zurich, Switzerland in 2016 and she is currently a Postdoctoral researcher at the Laboratory for Information and Decision Systems, M.I.T., USA.
Her research focuses on identification, analysis and control of complex systems, with application to distributed multi-agent networks, game theory and systems biology.
\end{IEEEbiography}
\begin{IEEEbiography}[{\includegraphics[width=1in,height=1.25in,clip,keepaspectratio]{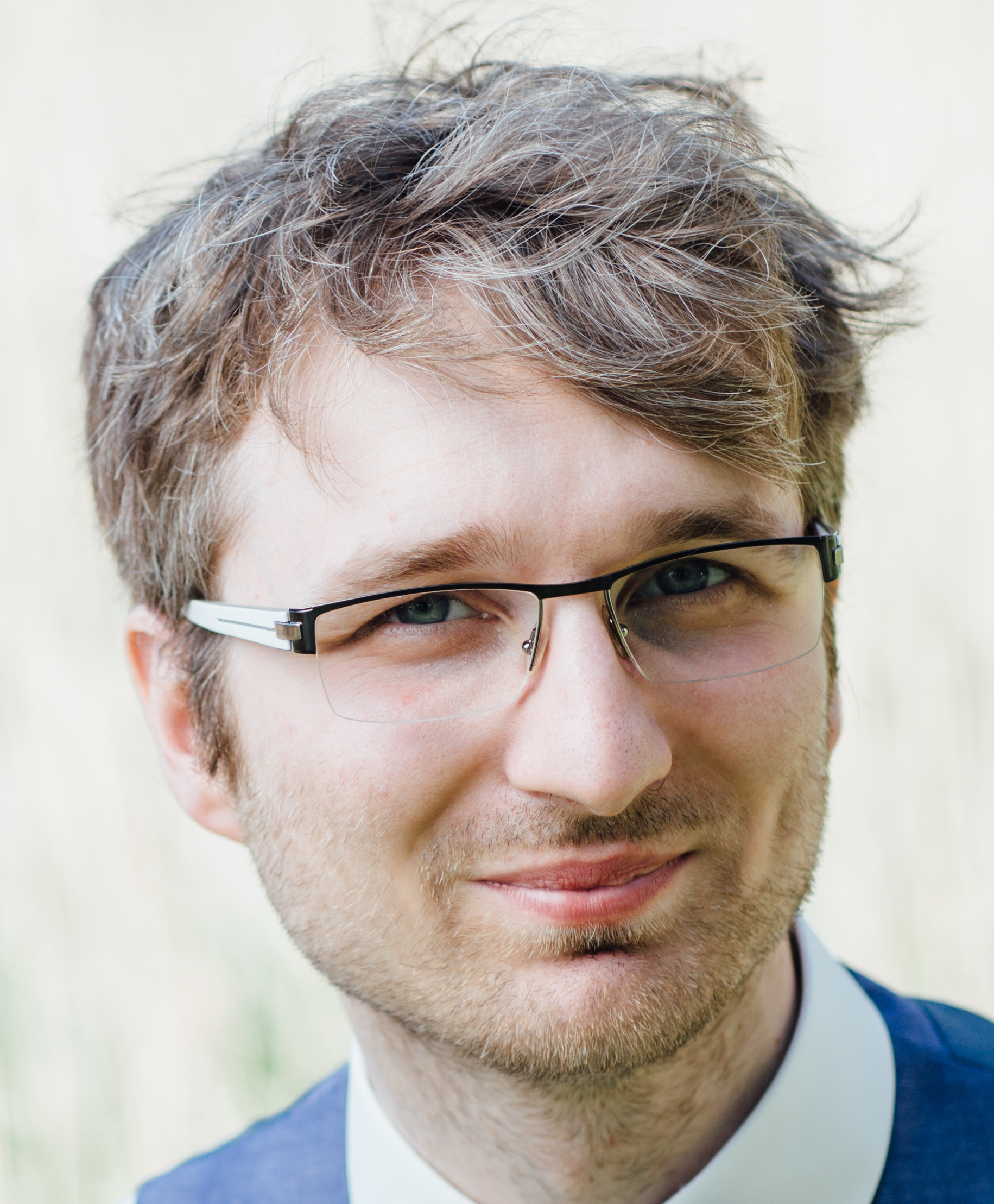}}]{Michael T. Schaub} obtained a B.Sc. in Electrical Engineering and Information Technology from ETH Zurich (2009), an M.Sc. in Biomedical Engineering  (2010) from Imperial College London, and a Ph.D. in Applied Mathematics (2014) from Imperial College. 
    After a Postdoctoral stay at the Universit\'e catholique de Louvain and the University of Namur (Belgium), he has been a Postdoctoral Researcher at the Institute of Data, Systems and Society (IDSS) at the Massachusetts Institute of Technology since November 2016. Presently, he is a Marie-Sklodowska Curie Fellow at IDSS and the Department of Engineering Science, University of Oxford, UK.
    He is broadly interested in interdisciplinary applications of applied mathematics in engineering, social and biological systems.
    His research interest include in particular network theory, data science, machine learning, and dynamical systems.
\end{IEEEbiography}
\begin{IEEEbiography}[{\includegraphics[width=1in,height=1.25in,clip,keepaspectratio]{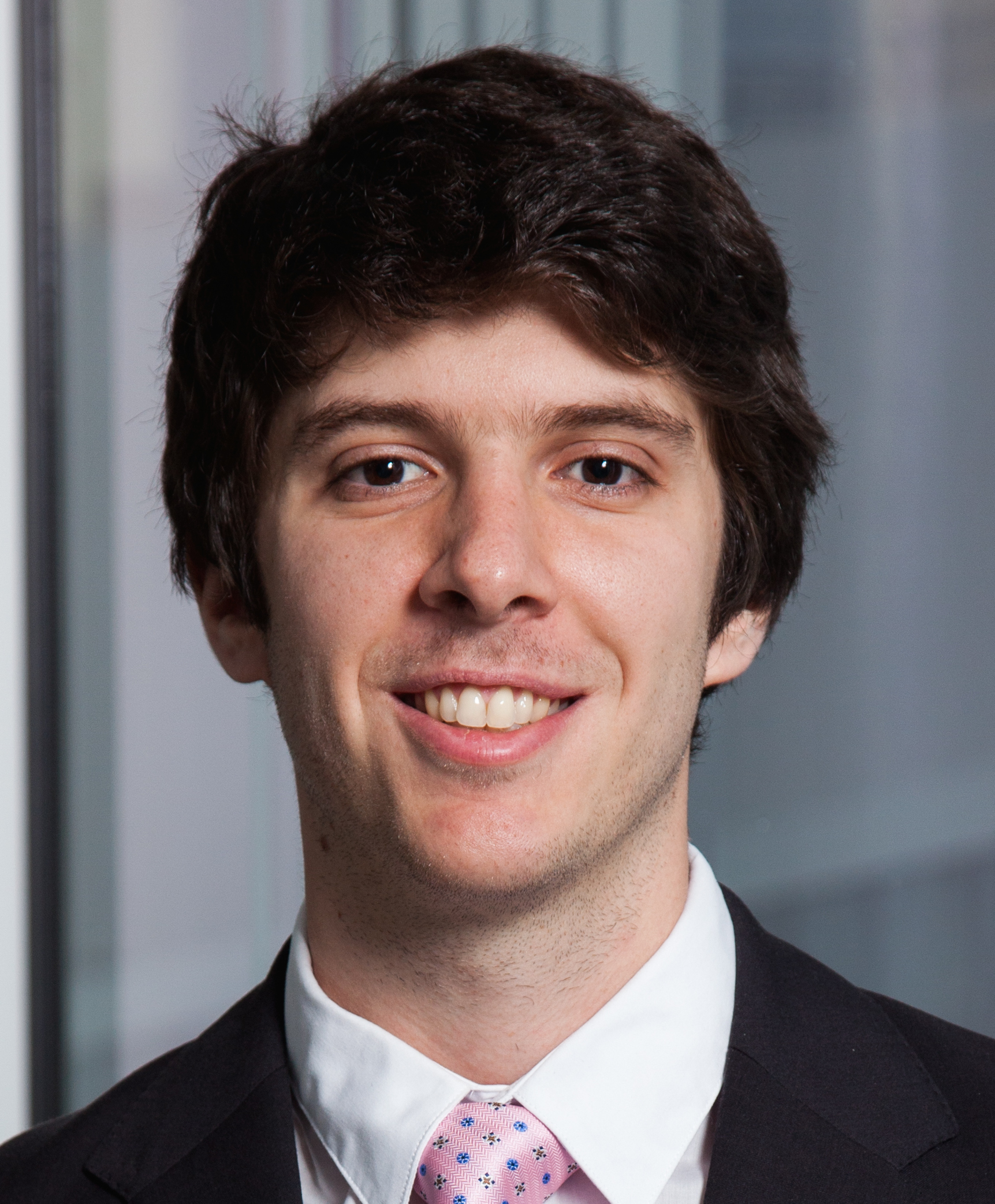}}]{Santiago Segarra} (M'16) received the B.Sc. degree in industrial engineering with highest honors (Valedictorian) from the Instituto Tecnol\'ogico de Buenos Aires (ITBA), Argentina, in 2011, the M.Sc. in electrical engineering from the University of Pennsylvania (Penn), Philadelphia, in 2014 and the Ph.D. degree in electrical and systems engineering from Penn in 2016. From September 2016 to June 2018 he was a postdoctoral research associate with the Institute for Data, Systems, and Society at the Massachusetts Institute of Technology. Since July 2018, Dr. Segarra is an Assistant Professor in the Department of Electrical and Computer Engineering at Rice University. His research interests include network theory, data analysis, machine learning, and graph signal processing. Dr. Segarra received the ITBA's 2011 Best Undergraduate Thesis Award in industrial engineering, the 2011 Outstanding Graduate Award granted by the National Academy of Engineering of Argentina, the 2017 Penn’s Joseph and Rosaline Wolf Award for Best Doctoral Dissertation in electrical and systems engineering as well as four best conference paper awards.
\end{IEEEbiography}

\cleardoublepage
\section*{Appendix B: Additional results}

\subsection*{Invariance of centrality measures under permutations}
Just as the topology of a graph is invariant with respect to relabelings or permutations of its nodes, graphons are defined only up to measure preserving transformations.
We show in the next lemma that the linear operator $\mathbb{W}^{\pi}$ associated with any such `permutation' $\pi$ (formalized via a measure preserving transformation) of a graphon $W$ shares the same eigenvalues of $\mathbb{W}$ and `permuted' eigenfunctions.

\setcounter{lemma}{5}
\begin{lemma}\label{measurepreserving}
    Consider the  graphon $ W^{\pi}(x,y):=W(\pi(x),\pi(y))$ obtained by transforming $W$ using the measure preserving function $\pi:[0,1]\to[0,1]$.
    Let $\mathbb{W}$ and $\mathbb{W}^{\pi}$ be the associated linear integral operators.
    If $(\lambda,\vf)$  is an eigenvalue-eigenfunction pair of $\mathbb{W}$, then $(\lambda,\vf \circ \pi)$  is an eigenvalue-eigenfunction pair of~$\mathbb{W}^{\pi}$.
\end{lemma}
\begin{proof}
From a direct computation we obtain that
\begin{align*}
&\textstyle(\mathbb{W}^{\pi}(\vf\circ\pi))(x)=\textstyle\int_0^1 W^{\pi}(x,y) \vf(\pi(y)) \mathrm{d}y \\&=\textstyle\int_0^1 {W}(\pi(x),\pi(y)) \vf(\pi(y)) \mathrm{d}y =\textstyle\int_0^1 {W}(\pi(x),y) \vf(y) \mathrm{d}y \\&=\textstyle  (\mathbb{W}\vf)(\pi(x))= \lambda \vf(\pi(x)).
\end{align*}
The third equality uses the fact that $\pi$ is a measure preserving transformation and the ergodic theorem \cite[Ch.~8]{dudley2002}.
\end{proof}
Lemma~\ref{measurepreserving} complements the discussion at the end of Section~\ref{Ss:prelim_graphons} by showing the effect of measure preserving transformations on the spectral properties of the graphon.

\subsection*{Discussion on related  graphon convergence results }
 We introduce some additional definitions in order to compare our work with previous results on graphon convergence. In particular, let us start by introducing the \emph{cut norm} which is typically  used for the statement of graphon convergence results. For a graphon $W$ in the graphon space $\mathcal W$, the cut norm is denoted by $\|W\|_{\Box}$ and is defined as
$$\|W\|_{\Box}:=\sup_{U,V}\bigg|\int_U\int_V W(x,y)\mathrm{d}x\mathrm{d}y\bigg|, $$
where $U$ and $V$ are measurable subsets of $[0,1]$. The \emph{cut metric} between two graphons $W,W'\in\mathcal W$ is
$$d_{\Box}(W,W'):=\inf_{\phi\in\Pi_{[0,1]}}\|W^{\phi}-W'\|_{\Box} ,$$
 where $W^{\phi}(x,y):=W(\phi(x,),\phi(y))$ and $\Pi_{[0,1]}$ is the class of measure preserving permutations $\phi:[0,1]\mapsto[0,1]$. Intuitively, the function $\phi$ performs a node relabeling to find the best match between $W$ and $W'$. Because of such relabeling, $d_{\Box}(W,W')$ is not a well defined metric in $\mathcal W$ since we might have that $d_{\square}(W,W')=0$ even if $W\neq W'$. To avoid such a problem, we define the space $\mathcal W $ as the space where we identify graphons up to measure preserving transformations, so that $d_{\Box}$ is a well defined metric in $\mathcal W$. It can be shown that the metric space $(W,d_{\Box})$ is complete \cite{Lovasz2012}.
 The following lemma is instrumental for our comparison to previous work, as it establishes the equivalence between the $L_{p,q}$ norms and the cut norm. Recall that for any $p,q\geq 1$, the  $L_{p,q}$ operator norm is defined as $\vertiii{\mathbb W}_{p,q}:=\sup_{f\in L_p,\|f\|_p=1}\|\mathbb W f\|_q$.
\begin{lemma}
\label{lem:normequivalence}
\cite{Janson2013} For any $W\in\mathbb W$ and all $p,1\in[1,\infty]$
$$\|W\|_{\Box}\leq\vertiii{W}_{p,q}\leq\sqrt{2}(4\|W\|_{\Box})^{\min(1-1/p,1/q)}. $$
In particular, for $p=q=2$ we get
$$|W\|_{\Box}\leq\vertiii{\mathbb{W}}\leq\sqrt{8\|W\|_{\Box}}. $$
\end{lemma}
 It follows from Lemma~\ref{lem:normequivalence} that  convergence of the graphon operator can  be deduced from previous works establishing  graphon convergence in cut norm. For example one can use the results of~\cite{Borgs2012, Szegedy2011} combined with Lemma~\ref{lem:normequivalence} to easily conclude that the graphon operator converges in operator norm. However taking this approach one typically does not obtain rates of convergence of the sampled graph's operator to the graphon operator.  A convergence rate is instead provided in~\cite[Lemma 10.16]{Lovasz2012} for general graphons. We next show that, for graphons satisfying Assumption \ref{lipschitz} with $K=0$, the result in~\cite[Lemma 10.16]{Lovasz2012} leads to  a slower rate of convergence than the one provided in Theorem~\ref{Thm1}.   
 More precisely, combining~\cite[Lemma 10.16]{Lovasz2012} and Lemma~\ref{lem:normequivalence},  we get that with probability at least $1-\exp(-\frac{N}{2\log N})$ it holds  
 \begin{equation}\label{eq:lov}\vertiii{\mathbb{S}_N-\mathbb{W}}\leq \sqrt{176}/(\log N)^{1/4}.\end{equation} 
 By defining $\delta=\exp(-\frac{N}{2\log N})$, the bound provided in Eq.  \eqref{eq:convergenceSN} (for  $\kappa_N=1$ and $K=0$) leads to 
 \begin{equation}\vertiii{\mathbb{S}_N-\mathbb{W}}\leq \mathcal{O}\left(1/(\log N)^{1/2}\right),\end{equation} 
 thus proving faster convergence.
Finally, we note that  the bounds provided in Theorem \ref{Thm1} are  not only tighter but also more flexible. In fact,  by  introducing the parameter $\delta$ we establish a trade off between sharper error bounds and the probability that such bounds hold, as typically done in concentration inequality results.
 
\subsection*{A useful variant of the Davis-Kahan theorem}

The following technical lemma is used to prove the convergence of the eigenvector centrality for graphons and is a consequence of the Davis-Kahan $\sin\theta$ theorem for compact operators in Hilbert space \cite{davisandkahan1970}.

\begin{lemma}  \label{lemma:dk_simple} 
Consider two linear integral operators $\mathbb{L}$ and $\hat{\mathbb{L}}$, with ordered eigenvalues $\{\lambda_k\}_{k\ge 1}$, $\{\hat\lambda_k\}_{k\ge 1}$. Let $\hat{\varphi}_1,\varphi_1$ be the eigenfunctions associated with the dominant eigenvalues $\hat \lambda_1$ and  $ \lambda_1$ (normalized to norm one) and suppose that $|\lambda_1-\lambda_2|>|\hat{\lambda}_1-\lambda_1|.$
Then 
\begin{equation}
\textstyle \|\hat{\varphi}_1-\varphi_1\|\leq \frac{\sqrt{2} \vertiii{\hat{\mathbb{L}}-\mathbb{L}} }{|\lambda_1-\lambda_2|-|\hat{\lambda}_1-\lambda_1|}.
\end{equation}
\end{lemma}

The proof can be found in~\cite{teamgraphon}.

\subsection*{Concentration of uniform order statistics}
The goal of this subsection is to derive a uniform deviation bound for order statistics sampled from a standard uniform distribution, as detailed in  Proposition \ref{concentration_all} in the main text. This result is  required in the proof of Theorem \ref{Thm1}. Although it is  intuitive to expect subgaussian deviations for  uniform order statistics, we could not find the desired statement explicitly in the literature and believe it could be of interest on its own right. From a technical point of view, the key ingredient in our argument is to use the exponential Efron-Stein inequality derived in \cite{boucheronandthomas2012}.

Let $U_1,\dots,U_N\sim Unif(0,1)$ and define their correspondent order statistics $U_{(1)}\leq U_{(2)}\leq \dots\leq U_{(N)}$ and spacings $\Delta_i=U_{(i)}-U_{(i-1)}$ for $i=1,\dots,N+1$ with the convention $U_{(i)}=0$ and $U_{(N+1)}=1$. It is shown in \cite{devroye1986}  that 
\begin{mylist}
\item[-] Each $U_{(i)}$ is  distributed according to $Beta(i,N+1-i)$ and thus has mean $\frac{i}{N+1};$
\item[-] The joint survival function of the spacings is
$$\mathbb P\big(\Delta_1>s_1,\dots,\Delta_{N+1}>s_{N+1}\Big)=\left(1-\sum_{i=1}^{N+1}s_i\right)^{N}_{+}.$$
Consequently the spacings are identically (but not independently) distributed with  cumulative distribution  $F_{\Delta}(s)= 1-(1-s)^{N}$ and  marginal density
$f_{\Delta}(s)=N(1-s)^{N-1}. $
\end{mylist}

The following lemma is a key intermediate step in the derivation of  concentration inequalities for order statistics drawn from a uniform distribution.

\begin{lemma}
\label{twosided}
For any $\lambda\geq 0$ 
it holds  $$\log\mathbb{E}e^{\lambda|U_{(i)}-\mathbb{E}U_{(i)}|}\leq\lambda\frac{N}{2}\mathbb{E}[\Delta_i(e^{\lambda\Delta_i}-1)]. $$
\end{lemma}
\begin{proof}
We show this result by first proving
\begin{enumerate}
\item[(a)] $\log\mathbb{E}e^{\lambda(U_{(i)}-\mathbb{E}U_{(i)})}\leq\lambda\frac{N-i+1}{2}\mathbb{E}[\Delta_i(e^{\lambda\Delta_i}-1)]$;
\item[(b)] $\log\mathbb{E}e^{\lambda(\mathbb{E}U_{(i)}-U_{(i)})}\leq\lambda\frac{i}{2}\mathbb{E}[\Delta_i(e^{\lambda\Delta_i}-1)] $.
\end{enumerate}
To this end, note that the hazard rate of a uniform distribution is increasing since it has the form $h(x)=\frac{1}{1-x}$. Therefore applying Theorem 2.9 of \cite{boucheronandthomas2012} shows (a). Note that therein the result is proven only for $i\ge N/2+1$ (equivalently in the notation of \cite{boucheronandthomas2012} $k:=N+1-i\le N/2$) but such condition on $i$ is never used in the proof. Indeed, in order to prove claim (2.1) of Theorem 2.9 in~\cite{boucheronandthomas2012} one only needs to show that
\begin{equation}
\label{prop2.3}
\mbox{Ent}[e^{\lambda X_{(k)}}]\leq k \mathbb{E}[e^{\lambda X_{(k+1)}}\psi(\lambda(X_{(k)}-X_{(k+1)}))] 
\end{equation}
for $1\leq k\leq N$, where $\mbox{Ent}[Y]=\mathbb{E}[Y\log Y]-\mathbb{E}[Y]\log\mathbb{E}[Y]$ is the entropy of a non-negative random variable $Y$. This follows easily from the arguments of Proposition 2.3 in \cite{boucheronandthomas2012}. Note that the authors only consider $k:=N+1-i\le N/2$ in this proposition because for $k>N/2$ the bound can be improved.

Let us now turn to the proof of $(b)$. Note that the beta distribution is reflection symmetric i.e.  if $X\sim Beta(\alpha,\beta)$ then $1-X\sim Beta(\beta,\alpha)$ for $\alpha,\beta>0$. Therefore $U_{(i)}\sim 1-U_{(N-i+1)}$ and $\mathbb{E}U_{(i)}-U_{(i)}\sim U_{(N-i+1)}-\mathbb{E}U_{(N-i+1)}$. Hence  by $(a)$ we have that 
\begin{align*}
&\log\mathbb{E}e^{\lambda(\mathbb{E}U_{(i)}-U_{(i)})} =\log\mathbb{E}e^{\lambda(U_{(N-i+1)}-\mathbb{E}U_{(N-i+1)})}\\
&\textstyle \leq \lambda\frac{N-(N-i+1)+1}{2}\mathbb{E}[\Delta_{N-i+1}(e^{\lambda\Delta_{N-k+1}}-1)] \\
 &\textstyle =\lambda\frac{i}{2}\mathbb{E}[\Delta_i(e^{\lambda\Delta_i}-1)],
\end{align*}
where in the last step we  used the fact that the spacings $\Delta_1,\dots,\Delta_N$ have the same marginal distribution. 
Finally, the statement of the lemma is an immediate consequence of (a) and (b).
\end{proof}

\begin{lemma}
 \label{concentration_abs}
Suppose that $N>5$ and $\delta_i\in(e^{-N/5},e^{-1})$  then, for $i=1,\dots,N$, with probability at least $1-\delta_i$
 $$\Big|U_{(i)}-\frac{i}{N+1}\Big|\leq \sqrt{\frac{8\log(1/\delta_i)}{(N+1)}}.$$
 \end{lemma} 
 \begin{proof}
We note that  by Chernoff's inequality
 \begin{equation}
 \label{markov}
 \mathbb P\Big(\Big|U_{(i)}-\frac{i}{N+1}\Big|>t\Big)\leq  \frac{\mathbb E[e^{\lambda|U_{(i)}-\frac{i}{N+1}|}]}{e^{\lambda t}}
 \end{equation}
 and from Lemma \ref{twosided} we see that
 \begin{equation}
 \label{boucheronthomas}
  \mathbb E [e^{\lambda|U_{(i)}-\frac{i}{N+1}| }] \leq e^{\lambda \frac{N}{2} \mathbb E[\Delta_i(e^{\lambda\Delta_i}-1)]}.
 \end{equation}
From the marginal density of the spacings we get 
$$\mathbb E[\Delta_i^k]\!=\!N \int_0^1 \!\! (s_i)^k (1-s_i)^{N-1} ds_i\! = \!N \mathcal{B}(k+1,N)\!= \! \frac{k! N!}{(N+k)!},$$
where we used the definition of the beta function $\mathcal{B}(x,y)= \int_0^1 t^{x-1} (1-t)^{y-1} dt =\frac{(x-1)! (y-1)!}{(x+y-1)!}$ for integers $x,y$.
Then 
\begin{align}
\label{crudebound}
 &\textstyle \mathbb{E}[\Delta_i(e^{\lambda\Delta_i}-1)]= \mathbb{E}[\Delta_i( \sum_{k=0}^\infty \frac{(\lambda \Delta_i)^k}{k!} -1)]\nonumber \\
 &\textstyle=  \sum_{k=1}^\infty \frac{\lambda^k}{k!} \mathbb{E}[\Delta_i^{k+1}]\nonumber
  = \sum_{k=1}^\infty \frac{\lambda^k(k+1) N!  }{(N+k+1)!} \nonumber\\
  &\textstyle\le   \sum_{k=1}^\infty \frac{\lambda^k(k+1) N!  }{N^{k} (N+1) N!} \nonumber=  \frac{1}{N+1} \sum_{k=1}^\infty \left(\frac{\lambda}{N}\right)^k(k+1)  \nonumber\\
  &\textstyle=  \frac{1}{N+1} \left[ \sum_{k=0}^\infty \left(\frac{\lambda}{N}\right)^k(k+1) -1 \right]\nonumber
         = \frac{1}{N+1} \left[ \frac{1}{(1-\frac{\lambda}{N})^2} -1 \right]\nonumber\\
    &\textstyle=  \frac{1}{N+1} \left[ \frac{1-(1-\frac{\lambda}{N})^2}{(1-\frac{\lambda}{N})^2} \right]\nonumber
                  =  \frac{1}{N+1} \left[ \frac{1-1-(\frac{\lambda}{N})^2+2\frac{\lambda}{N}}{(1-\frac{\lambda}{N})^2} \right]\nonumber\\
                  & \textstyle=  \frac{\lambda}{N(N+1)} \left[ \frac{2-\frac{\lambda}{N}}{(1-\frac{\lambda}{N})^2} \right]\nonumber\\
 \end{align}
 where we used $\sum_{k=0}^\infty \alpha^k (k+1)=\frac{1}{(1-\alpha)^2}$ (obtained by differentiating the geometric sum for $\frac \lambda N<1$).
 If we set $\lambda<0.35 N$ then  $$ \mathbb{E}[\Delta_i(e^{\lambda\Delta_i}-1)]\le  \frac{4\lambda}{N(N+1)}$$ since  $y<  \frac{7 - \sqrt{ 49 - 32}}{8} \approx 0.36$ implies $ \left[ \frac{2-y}{(1-y)^2} \right]<4$.
 Combining \eqref{markov}, \eqref{boucheronthomas} and \eqref{crudebound} yields
$$\mathbb P\Big(\Big|U_{(i)}-\frac{i}{N+1}\Big|>t\Big)\leq e^{\frac{2 \lambda^2}{N+1}}e^{-\lambda t}.$$
Minimizing over $\lambda$  leads to the choice $\lambda=\frac{t(N+1)}{4}$ and thus
$$\mathbb P\Big(\Big|U_{(i)}-\frac{i}{N+1}\Big|>t\Big)\leq  \exp\bigg(-\frac{t^2(N+1)}{8}\bigg).$$
The proof is concluded if we select $t=\sqrt{\frac{8\log(1/\delta_i)}{(N+1)}}$. Note that for this choice
$$\lambda\!=\! \frac{t(N+1)}{4}\!= \!\sqrt{\frac{8\log(1/\delta_i)}{(N+1)}} \frac{(N+1)}{4} \!=\! \sqrt{\frac{\log(1/\delta_i)(N+1)}{2} }.$$
We need to verify that   $\lambda<0.35N$ or  equivalently that $2(0.35)^2N^2 - \log(1/\delta_i)N - \log(1/\delta_i) >0 $. A sufficient condition is
$N >  \frac{\log(1/\delta_i) +\sqrt{ \log(1/\delta_i)^2 +8  \log(1/\delta_i) (0.35)^2}}{4 (0.35)^2}=:\bar N$.
Note that  
$ \bar N< \frac{ 1+\sqrt{ 1 +8   (0.35)^2}}{4 (0.35)^2} \log(1/\delta_i)<5\log(1/\delta_i)$, since $\log(1/\delta_i) >1$ for $\delta_i<e^{-1}$. Hence a simpler sufficient condition 
 is $N>5\log(1/\delta_i)$.

 \end{proof}

 \textbf{Proof of Proposition \ref{concentration_all}:}\\

 From Lemma \ref{concentration_abs} we known that for each $i=1,\dots,N$ if we set $\delta_i=\frac\delta N$ then with probability at least $1-\frac{\delta}{N}$ it holds
 $G_i:=\Big|U_{(i)}-\frac{i}{N+1}\Big|\leq \sqrt{\frac{8\log(N/\delta)}{(N+1)}}=:t.$
It then follows from the union bound that
$
\mathbb{P}\Big(\bigcap_{1\leq i\leq N}\{G_i\leq t\} \Big)=\Big(1-\mathbb{P}\Big( \bigcup_{1\leq i\leq N}\{G_i> t\} \Big) \Big)
  \geq 1- \sum_{i=1}^N \mathbb{P}\Big( G_i> t \Big)
\geq 1- \sum_{i=1}^N \frac{\delta}{N}=1-\delta.
$

\subsection*{Proof of Lemma \ref{lemma:max_degree}: A lower bound on the maximum expected degree}

Using the definition of $W_N$ and the reverse triangle inequality yields
\begin{align}
&\frac1NC^d_N=\frac1N\max_i\left( \sum_{j=1}^N P^{(N)}_{ij}   \right)=\frac1N\max_i\left( \sum_{j=1}^N W(u_i,u_j) \right  ) \notag\\
&=\max_{x\in[0,1]}\left(\int_0^1 W_N(x,y)dy\right) \ge \max_{x\in \mathcal C_N^c}\left(\int_0^1 W_N(x,y)dy\right) \notag \\
&\ge \max_{x\in \mathcal C_N^c}\left(\int_0^1 W(x,y)dy - \int_0^1 |D(x,y)|dy\right) \label{two_terms},
\end{align}
where $\mathcal C_N:=\{x\in[0,1]\mid \exists k\in\{1,\ldots,K\} \mbox{ s.t. } |x-\alpha_k|\le d_N \}$ is the subset of points in $[0,1]$ that are up to $d_N$ close to a discontinuity.
Note that  for any $x\in \mathcal C_N^c$, with probability $1-\delta'$ (see  part 1 of Theorem \ref{Thm1})
\begin{align*}
&\int_0^1 |D(x,y)|dy=\int_{C_N^c} |D(x,y)|dy +\int_{C_N} |D(x,y)|dy \\
&\le 2Ld_N+\textup{Area}(C_N) =2Ld_N+2Kd_N.
\end{align*}
Substituting in \eqref{two_terms} we get
$$\frac1NC^d_N \ge \max_{x\in \mathcal C_N^c}\left(\int_0^1 W(x,y)dy \right) -2(L+K)d_N.$$
Finally, note that Assumption \ref{lipschitz} implies that the degree $c^d(x)$ is piece-wise Lipschitz continuous, that is,  for any $k\in\{1,\ldots,K+1\}$ and any $x,x'\in \mathcal{I}_k$ it holds $|c^d(x)-c^d(x')|\le L|x-x'|.$ If $\Delta^{(\alpha)}_{\textup{MIN}}>2d_N$ this implies that
$$|\max_{x\in \mathcal C_N^c}\left(\int_0^1 W(x,y)dy \right)-C^d| \le Ld_N,$$
since there must be at least one point in  $\mathcal C_N^c$ which belongs to the same Lipschitz block as $\argmax c^d(x)$ and has distance $d_N$ from it. Overall, we have proven
\begin{align*}
\textstyle C^d_N \ge NC^d -N(3L+2K)d_N \ge \log\left(\frac{2N}{\delta} \right) \ge \frac49 \log\left(\frac{2N}{\delta} \right).
\end{align*}

\section*{Appendix C: Mathematical background}

For completeness, we provide a self-contained review of the mathematical tools required in the proofs of our results. The subsection on bounded linear operators is a condensed overview of concepts detailed in, e.g., \cite{conway2013, sauvigny2012b, deimling1985}. The subsection on perturbation theory introduces concepts necessary for a formal statement of the $\sin \theta$ theorem of \cite{davisandkahan1970} in the case of compact operators.

\subsection*{Bounded linear operators in Hilbert space}

Let us start by introducing some basic notions regarding linear operators in metric spaces.

\begin{definition}
Let $\mathcal{X}$, $\mathcal{Y}$ be normed linear spaces and let $\mathbb{L}:\mathcal{X}\to\mathcal{Y}$ be a linear operator.
\begin{mylist}
\item[(a)] $\mathbb{L}$ is continuous at a point $f\in \mathcal{X}$ if $f_n\to f$ in $\mathcal{X}$ implies $\mathbb{L}f_n\to \mathbb{L}f$ in $\mathcal{Y}$.
\item[(b)] $\mathbb{L}$ is continuous if it is continuous at every point, i.e. if $f_n\to f$ in $\mathcal{X}$ implies $\mathbb{L}f_n\to \mathbb{L}f$ in $\mathcal{Y}$ for every $f$.
\item[(c)] $\mathbb{L}$ is bounded if there exists a finite $M\geq 0$ such that, for all $f\in \mathcal{X}$,
$$\|\mathbb{L} f\|\leq M\|f\|. $$
Note that $\|L f\|$ is the norm of $L f$ in $\mathcal{Y}$, while $\|f\|$ is the norm of $f$ in $\mathcal{X}$.
\item[(d)] The operator norm of $\mathbb{L}$ is $\vertiii{\mathbb{L}}:=\sup_{\|f\|=1}\|\mathbb{L}f\|.$
\item[(e)] We let $\mathcal{B}(\mathcal{X,Y})$ denote the set of all bounded linear operators mapping $\mathcal{X}$ into $\mathcal{Y}$, that is
$$ \mathcal{B}(\mathcal{X,Y})=\{\mathbb{L}:\mathcal{X}\to\mathcal{Y}|~\mathbb{L}\mbox{ is bounded and linear}\}.$$
If $\mathcal{X}=\mathcal{Y}$ we write $\mathcal{B}(\mathcal{X})=\mathcal{B}(\mathcal{X,X})$
\end{mylist}
\end{definition}
The following Proposition shows that (a), (b) and (c) are equivalent.
\begin{proposition}
\label{equivboundcont}
 Let $\mathbb{L}:\mathcal{X}\to\mathcal{Y}$ be a linear operator. Then the following conditions are equivalent. 
 \begin{mylist}
 \item[(a)] $\mathbb{L}$ is continuous at every point of $\mathcal{X}$.
 \item[(b)] $\mathbb{L}$ is continuous at $0\in \mathcal{X}$.
 \item[(c)] $\|\mathbb{L}f\|$ is bounded on the unit ball $\{f\in \mathcal{X}; \|f\|\leq 1\}$.
\end{mylist}
 \end{proposition}

Let us now focus on linear operators acting on Hilbert spaces.
\begin{proposition} (Adjoint)
Let $\mathbb{L}\in\mathcal{B}(\mathcal{X,Y})$, where $\mathcal{X}$ and $\mathcal{Y}$ are Hilbert spaces. Then there exists a unique bounded linear map $\mathbb{L}^
{*}:\mathcal{Y}\to \mathcal{X}$ such that
$$ \langle \mathbb{L}x,y \rangle=\langle x,\mathbb{L}^*y \rangle \mbox{ for all } x\in \mathcal{X},~ y\in \mathcal{Y}.$$
\end{proposition}

\begin{definition}
Let $\mathcal{X}$ be a Hilbert space and $\mathbb{L}\in\mathcal{B}(\mathcal{X})$.
\begin{mylist}
\item[(a)] $\mathbb{L}$ is self-adjoint if 
$\mathbb{L}=\mathbb{L}^*$ i.e. $\langle \mathbb{L}x,y\rangle=\langle x,\mathbb{L}y\rangle$ for all $x,y\in \mathcal{X}$.
\item[(b)] $\mathbb{L}$ is compact if it maps the unit ball in $\mathcal{X}$ to a set with compact closure.
\end{mylist}

\end{definition}

 We are now ready to state the spectral theorem for compact operators.
\begin{theorem} (Spectral theorem, \cite[Theorem 2, Chapter 8 \S 7]{sauvigny2012b})
Let the $\mathbb{\mathbb{L}}:\mathcal{X}\to\mathcal{X}$ be a compact self-adjoint operator on the Hilbert space $\mathcal{X}$ satisfying $\mathbb{L}\neq 0$. Then we have a finite or countably infinite system of orthonormal elements $\{\vf_k\}_{k\geq 1}$ in $\mathcal{X}$ such that 
\begin{mylist}
\item[(a)] The elements $\vf_k$ are eigenfunctions associated with the eigenvalues $\lambda_k\in\mathbb{R}$, i.e.
$$\mathbb{L}\vf_k=\lambda_k\vf_k, ~~~k=1,2,\dots $$
If the set of nonzero eigenvalues is infinite, then $0$ is the unique accumulation point.
\item[(b)] The operator $\mathbb{L}$ has the representation 
$$\mathbb{L}=\sum_{k\geq 1}\lambda_k\varphi_k\varphi_k^* ~~\mbox{ i.e. }~~\mathbb{L} f=\sum_{k\geq 1}\lambda_k\langle \vf_k,f\rangle\vf_k ~~~\mbox{ for all  } f\in\mathcal{X}.$$
\end{mylist}
\label{thm:spec}
\end{theorem}

The following useful result shows that linear integral operators are compact.
\begin{proposition}
\label{compactness} (\cite[Chapter 2, Proposition 4.7]{conway2013})
If $K\in L^2([0,1]^2)$, then
$(\mathbb{K}f)(x)=\int_0^1 K(x,y)f(y)\mathrm{d}y $
is a compact operator. 
\end{proposition}

We conclude this subsection with a generalization of the Perron-Frobenius theorem to linear operators in Hilbert space. Let us  first introduce some additional notions used in the statement of the result. A closed convex set $K\subset\mathcal{X}$ is called a \emph{cone} if $\lambda K\subset K$ for all $\lambda\geq 0$ and $K\cap (-K)=\{0\}$. If the set $\{u-v: u,v\in K\}$ is dense in $\mathcal{X}$, then $K$ is called a \emph{total cone}.

\begin{theorem}(Krein-Rutman theorem, \cite[Theorem 19.2]{deimling1985}) \label{thm:kr}
Let $\mathcal{X}$ be a Hilbert space, $K\subset\mathcal{X}$ a total cone and $\mathbb{L}:\mathcal{X}\to\mathcal{X}$ a compact linear operator that is positive (i.e. $\mathbb{L}(K)\subset K$) with positive spectral radius $r(\mathbb{L})$. Then $r(\mathbb{L})$ is an eigenvalue with an eigenvector $\varphi\in K\setminus\{0\}:\mathbb{L}\varphi=r(\mathbb{L})\varphi$. 
\end{theorem}

\subsection*{Perturbation theory for compact self-adjoint operators}

The natural definition of the angle between two nonzero vectors $\varphi$ and $\tilde{\varphi}$ in a Hilbert space $\mathcal{X}$ is the number
\begin{equation}\label{eq:angle}
\Theta(\varphi,\tilde{\varphi})=\cos^{-1}\bigg(\frac{\langle \varphi,\tilde{\varphi}\rangle}{\|\varphi\|\|\tilde{\varphi}\|}\bigg).
\end{equation}
Note that the above concept is well defined because of Cauchy-Schwartz inequality.
Consider now the two subspaces spanned by the two nonzero vectors $\varphi$ and $\tilde{\varphi}$, that is
 $[\varphi]:=\varphi \varphi^*\mathcal{X}$ and $[\tilde{\varphi}]:=\tilde{\varphi}\tilde{\varphi}^*\mathcal{X}$ .  One can extend  \eqref{eq:angle} to define an angle between the two subspaces  $[\varphi]$ and $[\tilde{\varphi}]$ as
$$\Theta([\varphi],[\tilde{\varphi}]):=\inf_{u,v}\Big\{\Theta(u,v);u\in[\varphi], v\in[\tilde{\varphi}]\Big\} .$$
More generally, one can extend this definition of angle to subspaces spanned by eigenfunctions. This will be particularly useful in situations where we are interested in a compact self-adjoint operator $\mathbb{L}$ but we only have access to a modified operator $\widetilde{\mathbb{L}}=\mathbb{L}+\mathbb{H}$. Indeed, in this case one way to measure how close these operators are is to measure the angle between subspaces spanned by their eigenfunctions. Let us introduce some notation in order to formalize this. We write the subspace (eigenspace) spanned by the eigenfunctions $\{\varphi_k\}_{k=1}^m$ of $\mathbb{L}$ by $[\mathbb{E}_0]:=[\begin{matrix} \varphi_1 & \varphi_2 &\dots & \varphi_m \end{matrix}]$. We denote the projector of $[\mathbb{E}_0]$ by $\mathbb{P}_0=\mathbb{E}_0\mathbb{E}_0^*=\sum_{k=1}^m\varphi_k\varphi_k^*$ and its complementary projector by $\mathbb{P}_1=\mathbb{E}_1\mathbb{E}_1^*$. Now any vector $x\in\mathcal{X}$ can be written as 
$$x=\begin{pmatrix}\mathbb{E}_0 & \mathbb{E}_1\end{pmatrix}\begin{pmatrix} x_0 \\ x_1 \end{pmatrix}=\mathbb{E}_0x_0+\mathbb{E}_1x_1 ,$$
where $x_0=\mathbb{E}_0^*x$ and $x_1=\mathbb{E}_1^*x$. We therefore say that $x$ is represented by $\begin{pmatrix} x_0 \\ x_1 \end{pmatrix}$.  The corresponding notation for an operator $\mathbb{L}:\mathcal{X}\to\mathcal{X}$ is 
 $$\mathbb{L}=\begin{pmatrix}\mathbb{E}_0 &\mathbb{E}_1\end{pmatrix}\begin{pmatrix} \mathbb{L}_0 & 0 \\ 0 & \mathbb{L}_1 \end{pmatrix} \begin{pmatrix}\mathbb{E}_0^* \\ \mathbb{E}_1^*\end{pmatrix}=\mathbb{E}_0\mathbb{L}_0\mathbb{E}_0^*+\mathbb{E}_1\mathbb{L}_1\mathbb{E}_1^*, $$  
where $\mathbb{LE}_0=\mathbb{E}_0\mathbb{L}_0$ and $\mathbb{LE}_1=\mathbb{E}_1\mathbb{L}_1$.
Similarly, we can consider the eigenspace $[\mathbb{F}_0]$ spanned by the eigenfunctions $\{\tilde{\varphi}_k\}_{k=1}^m$ of $\mathbb{L+H}$ and write 
$$ \widetilde{\mathbb{L}}=\mathbb{L+H}=\mathbb{F}_0\widetilde{\mathbb{L}}_0 \mathbb{F}_0^*+\mathbb{F}_1\widetilde{\mathbb{L}}_1\mathbb{F}_1^*.$$  
The problem of measuring the closeness between the eigenspaces  $[\mathbb{E}_0]=\mathbb{P}_0\mathcal{X}$ and  $[\mathbb{F}_0]=\widetilde{\mathbb{P}}_0\mathcal{X}$ can be tackled by looking at the angle between these subspaces. To do so,  we can define a diagonal operator $\Theta_0$ using the principal angles between $\mathbb{E}_0$ and $\mathbb{F}_0$, i.e., $\begin{pmatrix} \cos^{-1}(s_1)&\dots&\cos^{-1}(s_m)\end{pmatrix}$ where $s_1\geq \dots\geq s_m$ are the singular values of $\mathbb{E}_0^*\mathbb{F}_0$, or equivalently the square-root of the nonzero eigenvalues of $\mathbb{E}_0^*\mathbb{F}_0\mathbb{F}_0^*\mathbb{E}_0$. Then, writing $\mathbb{S}:=\diag\begin{pmatrix} \cos^{-1}(s_1)&\dots&\cos^{-1}(s_m)\end{pmatrix}$, we can define $\Theta_0=\Theta(\mathbb{E}_0,\mathbb{F}_0)$ as
 $$\textstyle \Theta_0=\cos^{-1}(\mathbb{S}) ~\mbox{ i.e. }~\Theta_0 f=\sum_{k=1}^ m\cos^{-1}(s_k)\langle \phi_k,f\rangle\phi_k $$  for all $f\in\mathcal{X}$ and any basis $\{\phi\}_{k=1}^\infty.$
 We are now ready to state the Davis-Kahan $\sin \theta$ theorem.

\begin{theorem} (Davis-Kahan $\sin\theta$ theorem  \cite{davisandkahan1970})
Let $\mathbb{L}$ and $\widetilde{\mathbb{L}}=\mathbb{L+H}$ be two self-adjoint operators acting on the Hilbert space $\mathcal{X}$ such that $\mathbb{L}=\mathbb{E}_0\mathbb{L}_0\mathbb{E}_0^*+\mathbb{E}_1\mathbb{L}_1\mathbb{E}_1^*$ and $\mathbb{L+H}=\mathbb{F}_0\widetilde{\mathbb{L}}_0 \mathbb{F}_0^*+\mathbb{F}_1\widetilde{\mathbb{L}}_1\mathbb{F}_1^*$  with $[\mathbb{E}_0,\mathbb{E}_1]$ and $[\mathbb{F}_0,\mathbb{F}_1]$ orthogonal. If the eigenvalues of $\mathbb{L}_0$ are contained in an interval $(a,b)$, and the eigenvalues of $\widetilde{\mathbb{L}}_1$ are excluded from the interval $(a-\delta,b+\delta)$ for some $\delta>0$, then
 $$ \|\sin\Theta(\mathbb{E}_0,\mathbb{F}_0)\|=\|\mathbb{F}_1^*\mathbb{E}_0\| \leq  \frac{\|\mathbb{F}_1^*\mathbb{H}\mathbb{E}_0\|}{\delta} $$
for any unitarily invariant operator norm $\|\cdot\|$.
\label{thm:dk}
\end{theorem}
Note that the above theorem holds even for non-compact operators. Indeed,  one might consider more general orthogonal subspaces defined through their projectors, which in turn might not be written as countable sums of the product of the elements of an orthogonal basis \cite{davisandkahan1970}. 

\end{document}